\documentclass[11pt,draftclsnofoot,onecolumn,journal,letterpaper]{IEEEtran}
\usepackage{cite}
\usepackage{amsmath,amssymb,amsfonts}
\usepackage{algorithmic}
\usepackage{graphicx}
\usepackage{textcomp}
\usepackage{xcolor}
\def\BibTeX{{\rm B\kern-.05em{\sc i\kern-.025em b}\kern-.08em
    T\kern-.1667em\lower.7ex\hbox{E}\kern-.125emX}}

\newcommand{\ba}{\begin{array}}
\newcommand{\ea}{\end{array}}
\newcommand{\be}{\begin{displaymath}}
\newcommand{\ee}{\end{displaymath}}
\newcommand{\ben}{\begin{equation}}
\newcommand{\een}{\end{equation}}
\newcommand{\bena}{\begin{eqnarray}}
\newcommand{\eena}{\end{eqnarray}}
\newcommand{\beqa}{\begin{eqnarray*}}
\newcommand{\enqa}{\end{eqnarray*}}
\newcommand{\f}{\frac}
\newcommand{\bc}{\begin{center}}
\newcommand{\ec}{\end{center}}
\newcommand{\bi}{\begin{itemize}}
\newcommand{\ei}{\end{itemize}}
\newcommand{\benu}{\begin{enumerate}}
\newcommand{\eenu}{\end{enumerate}}
\newcommand{\bdes}{\begin{description}}
\newcommand{\edes}{\end{description}}
\newcommand{\bt}{\begin{tabular}}
\newcommand{\et}{\end{tabular}}

\newcommand \xibf{\mbox{\boldmath$\xi$\unboldmath}}

\newcommand \Sigmabf{\hbox{$\bf \Sigma$}}

\newcommand \abf{{\bf a}}
\newcommand \bbf{{\bf b}}

\newcommand \ebf{{\bf e}}

\newcommand \hbf{{\bf h}}

\newcommand \nbf{{\bf n}}

\newcommand \vbf{{\bf v}}
\newcommand \wbf{{\bf w}}
\newcommand \xbf{{\bf x}}
\newcommand \ybf{{\bf y}}

\newcommand \Abf{{\bf A}}

\newcommand \Ibf{{\bf I}}

\newcommand \Pbf{{\bf P}}

\newcommand \Wbf{{\bf W}}

\newcommand{\circlambda}{\mbox{$\Lambda$
             \kern-.85em\raise1.5ex
             \hbox{$\scriptstyle{\circ}$}}\,}

\usepackage{graphicx}
\usepackage{epstopdf}
\usepackage{amsmath,amssymb,amsthm,mathrsfs,amsfonts,dsfont}
\usepackage{epsfig}
\usepackage[ruled,linesnumbered]{algorithm2e}
\usepackage{color}
\usepackage{subfigure}
\usepackage{cite}
\usepackage{diagbox}
\usepackage{multirow}
\usepackage{float}
\usepackage{stfloats}
\usepackage{hyperref}
\usepackage{enumitem}
\usepackage{cases}
\usepackage{setspace}
\usepackage{adjustbox,lipsum}
\usepackage{comment}
\usepackage{lipsum}
\usepackage{mathtools}
\usepackage{tikz}
\usetikzlibrary{shapes,arrows}

\graphicspath{{figures/}}
\allowdisplaybreaks

\newcommand{\mypara}[1]{{\smallskip \noindent \bf #1}\hspace{0.1in}}
\newcommand{\ssf}[1]{\textrm{$\sf{#1}$}{}}

\newcommand{\alg}{\texttt{FLORAS} }
\newcommand{\algg}{\texttt{FLORAS}}

\newtheorem{theorem}{Theorem}

\newtheorem{lemma}{Lemma}
\newtheorem{assumption}{Assumption}

\newtheorem{defi}{Definition}

\newtheorem{remark}{Remark}

\DeclareMathOperator*{\argmin}{arg\,min}

\newcommand{\vect}[1]{\mathbf{#1}}
\newcommand{\avgvect}[1]{\mathbf{\overline{#1}}}
\newcommand{\expt}{\mathbb{E}}
\newcommand{\prob}{\mathbb{P}}

\newcommand{\norm}[1]{\left \| #1 \right \|}
\newcommand{\squab}[1]{\left [ #1 \right ]}
\newcommand{\dotp}[2]{\left \langle #1, #2 \right \rangle}

\ifodd 0

\else

\fi

\ifodd 0
\newcommand{\congc}[1]{{\color{magenta}(Cong: #1)}}
\else
\newcommand{\congc}[1]{}
\fi

\ifodd 0
\newcommand{\zixiang}[1]{{\color{blue}(Zixiang: #1)}}
\else
\newcommand{\zixiang}[1]{}
\fi

\ifodd 0
\newcommand{\zixiangb}[1]{{\color{blue}#1}}
\else
\newcommand{\zixiangb}[1]{#1}
\fi

\ifodd 0
\newcommand{\zixiangbb}[1]{{\color{blue}#1}}
\else
\newcommand{\zixiangbb}[1]{#1}
\fi

\ifodd 0
\newcommand{\zixiangbbb}[1]{{\color{blue}#1}}
\else
\newcommand{\zixiangbbb}[1]{#1}
\fi

\ifodd 0
\newcommand{\zixiangr}[1]{{\color{red}#1}}
\else
\newcommand{\zixiangr}[1]{#1}
\fi

\ifodd 0
\newcommand{\dprev}[1]{{\color{blue}#1}}
\else
\newcommand{\dprev}[1]{#1}
\fi

\begin{document}


% icc title
% \title{FLORAS: Differentially Private Wireless Federated Learning Using Orthogonal Sequences}

% journal title
\title{Differentially Private Wireless Federated Learning Using Orthogonal Sequences}

\author{Xizixiang~Wei \qquad Tianhao~Wang \qquad Ruiquan~Huang \qquad Cong~Shen \qquad Jing~Yang \qquad H.~Vincent~Poor
\thanks{A preliminary version of this work was presented at the 2023 IEEE International Conference on Communications \cite{wei2023icc}.}

\thanks{Xizixiang Wei and Cong Shen are with the Charles L. Brown Department of Electrical and Computer Engineering, University of Virginia, USA. (E-mail: \texttt{\{xw8cw,cong\}@virginia.edu}).

Tianhao Wang is with the Department of Computer Science, University of Virginia, USA. (E-mail: \texttt{tianhao@virginia.edu}).

Ruiquan Huang and Jing Yang are with The Department of Electrical Engineering, The Pennsylvania State University, USA. (E-mail: \texttt{\{rzh5514,yangjing\}@psu.edu}).

H. Vincent Poor is with the Department of Electrical and Computer Engineering, Princeton University, USA. (E-mail: \texttt{poor@princeton.edu}).
}}

\maketitle

\begin{abstract}
% v2
We propose a privacy-preserving uplink over-the-air computation (AirComp) method, termed \texttt{FLORAS}, for single-input single-output (SISO) wireless federated learning (FL) systems. From the perspective of communication designs, \texttt{FLORAS} eliminates the requirement of channel state information at the transmitters (CSIT) by leveraging the properties of orthogonal sequences. From the privacy perspective, we prove that \texttt{FLORAS} offers both \emph{item-level} and \emph{client-level} differential privacy (DP) guarantees. Moreover, by properly adjusting the system parameters, \texttt{FLORAS} can flexibly achieve different DP levels at no additional cost. A new FL convergence bound is derived which, combined with the privacy guarantees, allows for a smooth tradeoff between the achieved convergence rate and differential privacy levels. Experimental results demonstrate the advantages of \texttt{FLORAS} compared with the baseline AirComp method, and validate that the analytical results can guide the design of privacy-preserving FL with different tradeoff requirements on the model convergence and privacy levels.
% v1
% A novel uplink over-the-air compuation (AirComp) design specific for single-input single-output (SISO) federated learning (FL) systems is proposed in this paper. The simple-yet-effective design provides both item-level and client-level differential privacy (DP) guarantee. Leveraging the properties of orthogonal sequences, the proposed design does not require channel state information at transmitter (CSIT). Moreover, by the adjustment of system parameter, the proposed design can flexibly achieve different DP levels at no additional cost. An interesting convergence bound is then derived, which trade off convergence speed and differential privacy levels of the proposed design. Extensive experiments based on real-world datasets are conducted to evaluate the proposed method. Numerical results demonstrate the advantage of the proposed design compared with state-of-arts AirComp method and validate that our analytical results are helpful for the design on privacy-preserving FL architectures with different tradeoff requirements on convergence performance and privacy levels.

\end{abstract}


\section{Introduction}
\label{sec:intro}
Real-world data generated or collected by edge devices enables various machine learning (ML) applications. For certain privacy-sensitive tasks, users prefer to keep their data locally instead of uploading to cloud servers. Federated learning (FL) \cite{mcmahan2017fl,konecny2016fl} has emerged as a distributed learning paradigm that caters to this growing trend, and is able to train a global ML model across all local datasets without the server having direct access to client data. 
 
The local training nature of FL leads to massive communication costs, as an FL task consists of multiple learning rounds, each of which requires uplink and downlink model exchange between clients and the server. Compared with downlink broadcasting, uplink communication is more challenging in FL when communication is over the wireless medium \cite{wei2023twc,Zheng2020jsac,mu2022isit}. Due to the stringent power constraints at mobile devices, channel noise and fading have a more conspicuous impact on uplink communications. More importantly, significant scalability challenges arise from the large number of clients in FL versus limited uplink communication resources. Uplink communication is known to be one of the key bottlenecks of wireless federated learning \cite{mcmahan2017fl}.

To tackle the scalability problem, over-the-air computation (also known as \emph{AirComp}) mechanisms have been proposed. %; see \cite{niknam2020federated} and the references therein. 
Instead of decoding individual local models of clients and then aggregating, AirComp allows multiple clients to transmit uplink signals in a naturally superpositioned fashion over a wireless medium, and decodes the average global model directly at the FL server. AirComp dramatically improves the scalability of wireless FL, and reduces the signal processing latency. Therefore, AirComp is regarded as a key technology for FL in wireless networks and has been investigated extensively \cite{zhu2019broadband,yang2020federated,amiri2020federated,cao2020tpc}. %; see Section~\ref{sec:related} for a literature review.
 
The most popular AirComp method is based on channel inversion power control \cite{zhu2019broadband}, which ``inverts'' the fading channel at each transmitter so that the aggregated model can be directly obtained at the server. Variants and enhancements of AirComp have been studied, and a literature review can be found in Section~\ref{sec:related-A}. Yet, a fundamental limitation of the existing methods is that they mostly require channel state information at the transmitter (CSIT). Enabling CSIT in wireless communication systems is complicated and is substantially harder than obtaining the channel state information at the receiver (CSIR).  Moreover, channel inversion based on CSIT is well known to ``blow up'' when one of the users' channels is in deep fade \cite{TV:05}. Hence, exploring CSIT-free AirComp methods becomes attractive\cite{zhu2021over}. 

Meanwhile, amidst a growing focus on data security, the importance of safeguarding individuals' personal information has become increasingly emphasized. Although FL intuitively helps protect client privacy by keeping training data locally and never sharing it with the server, private information can still be leaked to some extent by analyzing the ML model parameters trained and uploaded by the clients \cite{shokri2015privacy,wang2019beyond,ma2020safeguarding}. To address the privacy concern, a natural way is to add (artificial) noise to ML model parameters in the upload phase of FL, whose privacy properties can be mathematically characterized using differential privacy (DP) \cite{dwork2014algorithmic}. %Traditionally, adding artificial noise to achieve certain differential privacy during communications in FL will trade off with its convergence \cite{geyer2017differentially,truex2019hybrid}. 

AirComp has the potential to provide DP guarantee at no extra cost due to the inherent natural noise in the wireless channel. Heuristically, different DP levels can be guaranteed by controlling the signal-to-noise ratio (SNR), and thus the effective channel noise level, at the receiver side. Yet, most of the literature on AirComp rarely characterizes the achievable privacy in a mathematically rigorous fashion. 
% Wei et al. \cite{wei2020federated} propose an AirComp design to achieve certain DP by adjusting proper effective noise. It also theoretically analyzes the convergence behavior with a differential privacy guarantee. Seif et al. \cite{seif2020wireless} maximize the convergence rate while satisfying a desired privacy level by optimizing the power allocation between local gradients and the artificial noise. 
A literature review can be found in Section \ref{sec:related-B}. Moreover, we note that the existing literature focuses only on item-level DP in wireless FL, while client-level DP (also known as user-level DP) \cite{levy2021learning} is a new metric that is particularly worth investigating for FL, which is largely missing.

To simultaneously remove the CSIT requirement of AirComp and address the privacy challenge, we propose \alg -- \underline{F}ederated \underline{L}earning using \underline{OR}thogon\underline{A}l \underline{S}equences, a novel uplink wireless physical layer design for FL by leveraging the properties of \emph{orthogonal sequences}. {On the communication design,} \alg preserves all the advantages of AirComp while removing the CSIT requirement. From the perspective of privacy, \alg achieves desired DP guarantees (both item-level and client-level) by adjusting the number of used orthogonal sequences, making it much simpler and providing more flexibility to trade off privacy and utility. %and obtaining different DP levels flexibly via the randomness from natural channel noise. 

The main contributions of this paper are summarized as follows:
\begin{itemize}[leftmargin=*] \itemsep=0pt
    \item We propose \alg for uplink communications in single-input single-output (SISO) wireless FL systems. \alg enjoys all the advantages of AirComp, yet without the CSIT requirement. In particular, orthogonal sequences allow the base station (BS) to obtain  individual CSIR via a single pilot, by which global ML model parameters can be estimated via simple linear projections. Therefore, \alg significantly reduces the channel estimation overhead. Different from the channel inversion power control, \alg allows the transmit power to be independent of the channel realizations, which avoids increasing the dynamic range of the transmit signal and improves the power efficiency.
    \item By adjusting the number of orthogonal sequences in the system configuration, the novel signal processing technique in \alg produces Cauchy effective noise to the decoded global model, which empowers flexible item-level and client-level DP guarantees. Moreover, a new FL convergence bound based on the truncated Cauchy noise is derived, which allows us to characterize the tradeoff between the model convergence rate and the achievable DP levels. %To the best of the authors’ knowledge, this is the first piece of work that discusses pure differential privacy in AirComp of FL. 
    \item We conduct extensive experiments based on real-world datasets to evaluate the performance of \algg. Numerical results demonstrate the performance advantages of \alg compared with the channel inversion method and validate our theoretical analysis by achieving tradeoffs between model convergence and DP. %Therefore, our analytical results are helpful for the design on privacy-preserving FL architectures with different tradeoff requirements on convergence performance and privacy levels.
\end{itemize}

The remainder of this paper is organized as follows. Literature review is presented in Section \ref{sec:related}. Section \ref{sec:model} introduces the FL pipeline and the uplink communication model. The proposed \alg design is detailed in Section \ref{sec:propMtd}. The DP guarantee and convergence analysis of \alg are presented in Section \ref{sec:privacy} and Section \ref{sec:CvgAna}, respectively. Experimental results are reported in Section \ref{sec:exp}, followed by the conclusion of our work in Section \ref{sec:concl}.

\section{Related Work}\label{sec:related}
%\mypara{Improve FL communication efficiency.} The original \textsc{FedAvg} reduces the communication overhead by only periodically averaging the local models. Theoretical understanding of the communication-computation tradeoff has been actively pursued and, depending on the underlying assumptions (e.g., IID or non-IID local datasets, convex or non-convex loss functions, GD or SGD), rigorous analyses of the convergence behavior have been carried out \cite{stich2018local,wang2018cooperative,li2019convergence}. The approaches to reduce the size or frequency of including sparsification\cite{thonglek2022sparse,hu2022fedsynth,zhu2022byzantine} and quantization\cite{zhu2020one,amiri2020federated2,du2020high,oh2022fedvqcs}. There are also recent efforts in improving resource allocation efficiency to reduce the communication cost \cite{wen2021adaptive, chen2022federated,wang2021federated}. Nevertheless, mostly of these works consider interference-and-error-free rate-limited communication links and ignore the physical characteristics of wireless communication channels.

\subsection{AirComp for FL}\label{sec:related-A}
%AirComp \cite{zhu2019broadband,yang2020federated,amiri2020federated,cao2020tpc} leverages the natural signal superposition properties in a wireless multiple access channel to efficiently compute a sum/average function. It can be viewed as a special case of computing over multiple access channels \cite{nazer2007it}. This technique has attracted a lot of interest as it can reduce the uplink communication cost to be (nearly) agnostic to the number of participating clients. Client scheduling and various power and computation resource allocation methods have been investigated \cite{chen2020convergencenew,xu2020client,sun2021dynamic,lee2021adaptive, wadu2021joint}. Full CSIT is relaxed in \cite{sery2020tsp} by only using the phase information of the channel. Several studies have provided convergence guarantees of AirComp under different practical constraints and different types of heterogeneity \cite{lin2021deploying,aygun2022over, wan2021convergence,sery2021over,sun2022time}. In particular, \cite{amiri2021blind,wei2023twc} propose CSIT-free AirComp methods leveraging channel orthogonality. However, they are effective only in massive MIMO systems.

The AirComp approach \cite{zhu2019broadband,yang2020federated,amiri2020federated,cao2020tpc} exploits the inherent signal superposition characteristics of a wireless multiple access channel to efficiently perform sum/average computations. This methodology can be considered as a special instance of computation over multiple access channels, as outlined by \cite{nazer2007it}. The approach has garnered significant attention due to its ability to minimize uplink communication costs, irrespective of the number of participating clients in FL. The exploration of client scheduling, along with various power and computation resource allocation techniques, has been investigated by \cite{chen2020convergencenew,xu2020client,sun2021dynamic,lee2021adaptive,wadu2021joint}. Several studies have provided convergence guarantees for AirComp under diverse practical constraints and types of heterogeneity  \cite{lin2021deploying,aygun2022over,wan2021convergence,sery2021over,sun2022time}. Efforts also have been made to reduce the CSIT requirement of AirComp. \cite{sery2020tsp} relaxes full CSIT by utilizing only the phase information of the channel. Notably, \cite{amiri2021blind} and \cite{wei2023twc} present CSIT-free AirComp methods that leverage channel orthogonality, although their effectiveness is limited to massive MIMO systems.

\subsection{Differential Privacy for FL} \label{sec:related-B}
\zixiangb{DP-SGD has been widely regarded as a standard approach to train a differentially private ML model \cite{abadi2016deep}. Along with its variants\cite{du2021dynamic,van2022generalizing}, recent years have also witnessed increased efforts on DP for distributed learning systems, including the clipping technique in \cite{zhang2022understanding,andrew2021differentially}, the subsampling principle in \cite{balle2018privacy}, and random quantization in \cite{agarwal2018cpsgd}.} In the category of exploiting the channel noise for differentially private FL, 
\cite{wei2020federated} proposes an AirComp design to achieve DP by adjusting the effective noise; \cite{seif2020wireless,liu2020privacy} investigate adding noise and power allocation in non-orthogonal multiple access (NOMA) systems; \cite{mohamed2021privacy} considers DP amplification via user sampling and wireless aggregation; \cite{hu2020personalized} applies it to personalized FL. \cite{wei2021low} jointly optimizes the latency and DP requirements of FL, and \cite{kim2021federated} discusses the tradeoff among privacy, utility, and communication.

% \begin{figure}
%     \centering
%     \includegraphics[width = 0.9\linewidth]{flpipeline.png}
%     \caption{Illustration of a wireless FL system.}
%     \vspace{-0.1in}
%     \label{fig:FLpipeline}
% \end{figure}

\section{System Model}
\label{sec:model}
\subsection{FL Model}
Consider an FL task with a parameter server and $M$ clients. Each client $k \in [M]$ stores a (disjoint) local dataset $\mathcal{D}_k$, with its size denoted by $D_k$. The amount of the total data is $\zixiangr{D_{\text{total}}} \triangleq \sum_{k\in [M]} D_k$. We use $f_k(\wbf)$ to denote the local loss function at client $k$, which measures how well an ML model with parameter $\wbf \in \mathbb{R}^d$ fits its local dataset. The global objective function over all $M$ clients can be expressed as
\begin{equation*}
    f(\wbf) = \sum_{k \in [M]} p_k f_k(\wbf),
\end{equation*}
where $p_k = \frac{D_k}{\zixiangr{D_{\text{total}}}}$ is the weight for each local loss function, and the goal of FL is to find the optimal model parameter $\wbf^*$ that minimizes the global loss function: 
\begin{equation*}
    \wbf^* \triangleq \argmin_{\wbf\in\mathbb{R}^d}f(\wbf).
\end{equation*} 
We define $\Gamma \triangleq f^* - \sum_{k \in [M]} p_k f_k^*$ to capture the degree of the non-independent and identically distribution (non-IID) of local datasets \cite{li2019convergence}, where $f^*$ and $f^*_k$ are the minima of global and local loss functions, respectively.

The $\textsc{FedAvg}$ framework \cite{mcmahan2017fl} keeps client data locally, and the global model is obtained at the parameter server by the composition of multiple learning rounds. One of the key characteristics of FL is \emph{partial clients participation}, i.e., only a portion of clients are selected in a single learning round for model upload. Here, we assume that $K$ of total $M$ clients are uniformly randomly selected during each learning round for the FL task. To simplify the notation, we use the subscript $k = 1,\cdots, K$ to indicate the participating $K$ clients during a given learning round, acknowledging that they could correspond to different clients in different rounds. 

% A typical wireless FL pipeline is illustrated in Fig.~\ref{fig:FLpipeline}. Specifically, this pipeline iteratively executes the following steps in the $t$-th learning round.%, $t=1, \cdots, T$.

A typical wireless FL pipeline iteratively executes the following steps in the $t$-th learning round.
\begin{enumerate}
\item \textbf{Downlink wireless communication.} The BS broadcasts the current global model $\wbf_t$ to all $K$ selected devices over the downlink wireless channel.
\item \textbf{Local computation.} Each client $k$ uses its local data to train a local ML model $\wbf_{t+1}^k$ improved upon the received global model $\wbf_t$. We assume that mini-batch stochastic gradient descent (SGD) is adopted to minimize the local loss function. The parameter is updated iteratively (for $E$ steps) at client $k$ as: 
% $\wbf_{t,0}^k  = \wbf_t; \wbf_{t,\tau}^k  = \wbf_{t,\tau-1}^k - \eta_t \nabla\tilde{f}_k(\wbf_{t,\tau - 1}^k,\xi_{t,\tau - 1}^k), \forall \tau = 1, \cdots, E;  \wbf_{t+1}^k  = \wbf_{t,E}^k$
\begin{align*}
&\wbf_{t,0}^k  = \wbf_t, \\
&\wbf_{t,\tau}^k  = \wbf_{t,\tau-1}^k - \eta_t \nabla\tilde{f}_k(\wbf_{t,\tau - 1}^k,\xi_{t,\tau - 1}^k), \forall \tau = 1, \cdots, E,\\
% &\qquad\qquad\forall \tau = 1, \cdots, E;\\
& \wbf_{t+1}^k  = \wbf_{t,E}^k,
\end{align*} 
where $\nabla\tilde{f}_k(\wbf, \xi)$ denotes the stochastic gradient at client $k$ on model $\wbf$ and mini-batch $\xi$.
\item \textbf{Uplink wireless communication.} Each involved client uploads its latest local model to the server synchronously over the uplink wireless channel.
\item \textbf{Server aggregation.} The BS aggregates the received noisy local models $\tilde \wbf_{t+1}^k$ to generate a new global model. For simplicity, we assume that each local dataset has an equal size, {i.e., $D_k = D, \forall k \in [M]$}; therefore we have $\wbf_{t+1} = \Sigma_{k=1}^K \frac{1}{K} \tilde \wbf_{t+1}^k$.
\end{enumerate}

This work focuses on steps $3$ and $4$ in the FL pipeline. In particular, we leverage the unique properties of orthogonal sequences, which leads to an efficient FL uplink communication design with DP guarantees.

\subsection{Communication Model}
\label{sec:commmodel}
In each learning round, since there are $K$ active clients, the uplink communication between clients (mobile devices) and the parameter server (base station) can be modeled as over a fading multiple access channel. Consider a cell with a single-antenna BS and $K$ single-antenna users involved in one round of the aforementioned FL task. The communication system leverages orthogonal sequences for uplink transmissions. Note that one of the most popular implementations of the orthogonal sequence-based system is code-division multiple access (CDMA). We assume a spreading sequence set $\mathcal{A}=\{ \abf_1, \cdots, \abf_k, \cdots, \abf_N \}$ containing $N$ unique spreading sequences ($N\geq K$), where each spreading sequence is denoted as $\abf_k = [a_{1,k},\cdots, a_{L,k} ]^T$ and $L$ is the length of the spreading sequence. Each user is (randomly) assigned with a unique spreading sequence $\abf_k$ from $\mathcal{A}$ as its signature. 

We assume that the BS only has knowledge of the entire spreading sequence set $\mathcal{A}$, \emph{without knowing the specific signature of each user}. We emphasize that this restriction is consistent with our goal of guaranteeing user privacy -- BS cannot identify users or decode individual models based on their spreading sequences. We will discuss the details of the spreading sequence assignment mechanism in Section \ref{subsec:AgnmntMchsm}.

In the uplink communication, each client transmits the differential  between the received global model and the computed new local model: 
$$\zixiangr{{\bf{x}}_t^k} = \wbf_t - \wbf_{t+1}^k\zixiangr{\in \mathbb{R}^{d}},\;\; \forall k = 1,\cdots,K.$$
The BS aims at estimating $\zixiangr{\xbf_t}\triangleq \sum_{k = 1}^K {\bf{x}}_t^k$. 

Before the transmission of \zixiangr{${\xbf}_t^k$}, client $k$ \zixiangr{will apply a normalization technique. We denote $${\mathsf{x}_t^k} \triangleq [x_{1,t}^k, \cdots, x_{i,t}^k,\cdots, x_{d,t}^k]^T \in\mathbb{R}^d$$ as the normalized transmit signal. The following normalization technique, adopted in \cite{zhu2019broadband}, ensures $\expt[\mathsf{x}_t^k]=0$ and {$\norm{{\mathsf{x}_t^k}}_2^2 \leq C^2$}:
\begin{equation*}
    \mathsf{x}_t^k \triangleq \frac{C ({\xbf}_t^k - \mu_k)}{{\zixiangr{C_{\max,t}}}},
\end{equation*}
where $\mu_k$ is the sample mean of $d$ elements in ${\xbf}_t^k$. Note that normalization parameters $\mu_k$, $C$ and $\zixiangr{C_{\max,t}} \triangleq \max\{\norm{{\xbf}_t^k - \mu_k}_2,\forall k\}$ will be determined by the BS and clients via a separate control channel as suggested in \cite{zhu2019broadband}.} The $l_2$-norm bound guaranteed in the normalization not only provides the sensitivity of $\mathsf{x}_t^k$ for the DP analysis in Section \ref{sec:privacy}, but also satisfies the practical requirement that each client has a limited transmit power. We note that the similar technique has also been applied in \cite{zhang2022understanding}.

To simplify the notation, we omit index $t$ and use $x_k^i$ instead of $x_{i,t}^k$ barring any confusion. We assume that each client transmits every element of the model differential $\{x_k^1,\cdots,x_k^d\}$ via $d$ shared time slots. %\footnote{In general, differential model parameters can be transmitted over any $d$ shared time-frequency resources. For simplicity, we use $d$ time slots here.} 
In addition, \emph{block fading channel} is assumed\footnote{The large-scale pathloss and shadowing effect is assumed to be taken care of by, e.g., open loop power control \cite{SesiaLTE}, which is a common practice in real-world systems.}, i.e., the fading channel between each client and the BS $h_k$ remains unchanged in $d$ time slots. We emphasize that we do not make any specific assumption on the fading distribution throughout this paper. 
In the $i$-th slot, each client transmits symbol $x_{k}^i$ spread by its uniquely assigned orthogonal sequence $\abf_k$. The BS received signal can be written as%\footnote{For simplicity, we assume real signals $\{x_{k}^i\}_{i = 1}^d$ are transmitted in this paper. It can be easily extended to complex signals by stacking two real model parameters into a complex signal, so that the full d.o.f. is utilized.}
\begin{align*}
    \ybf_i = \sum_{k = 1}^K \abf_k  h_k x_k^i + \nbf_i, \;\; \forall i = 1,\cdots,d,
\end{align*}
where $\nbf_i$ is the additive white Gaussian noise (AWGN) with mean zero and variance $\sigma^2/L$ per dimension. Note that since the model differential parameters are real signals, we only need to consider the real parts of channel coefficients and noise. Although one-dimensional (real) modulation cannot fully leverage the channel degrees of freedom, it is consistent with the fact that binary phase-shift keying (BPSK) is the most common modulation scheme in CDMA systems \cite{TV:05}. In addition, focusing on the real dimension makes the subsequent discussion easier and highlights our contribution better. Also, as we detail later in Remark \ref{remark:channelgain}, we can leverage full channel gain at an affordable cost of partial CSIT. 

We note that spreading sequences are \emph{orthonormal}, i.e.,
\begin{equation}\label{eq:orthogonalSC}
    \abf_i^T\abf_i = 1,\;\;\forall i; \quad \text{and}\quad\abf_i^T\abf_j = 0, \;\;\forall i\neq j.
\end{equation}
At the BS, the receiver will decode the estimated aggregation parameter $\tilde{x}_i$, which is a noisy version of $x_i \triangleq \sum_{k=1}^K x_{k}^{i}$, \zixiangr{and recover ${ \mathsf{\tilde x}}_t \triangleq \squab{\tilde{x}_1, \cdots,\tilde{x}_d}^T$} in $d$ slots. After that, the BS can perform de-normalization:
\begin{equation*}
    \tilde \xbf_t \triangleq \frac{\zixiangr{C_{\max,t}}}{C}\zixiangr{\mathsf{\tilde x}_t} + \sum_{k = 1}^K \mu_k
\end{equation*}
and compute the new global model as
\begin{equation}\label{eq:diffGlobal}
    \wbf_{t + 1} = \wbf_t + \frac{1}{K}\tilde \xbf_t.
\end{equation}

Throughout the paper, we assume that all users are synchronized in frames, which can be achieved by the BS sending a beacon signal to initialize uplink transmissions \cite{SesiaLTE}.
\begin{figure*}
    \centering
    \includegraphics[width = 0.95\linewidth]{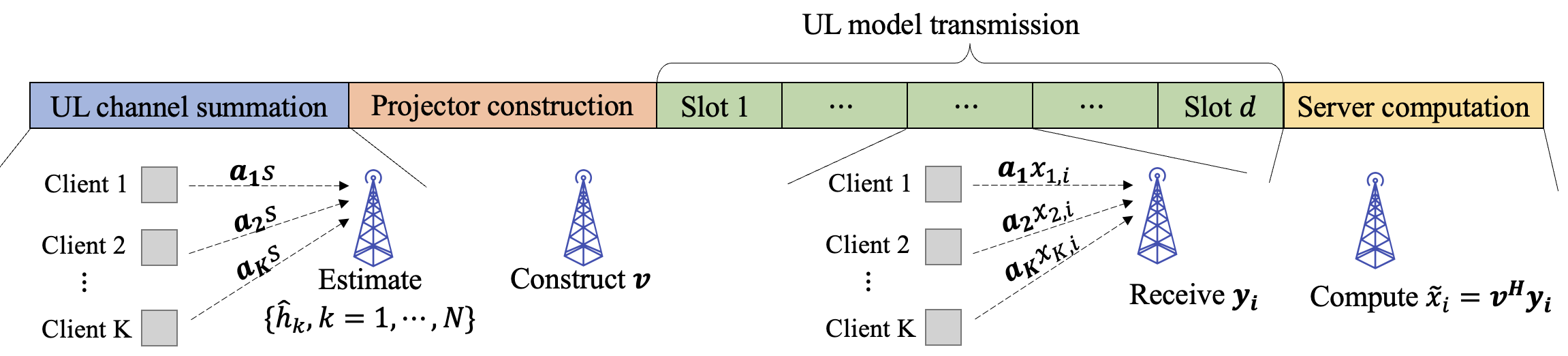}
    \caption{Illustration of the proposed uplink communication design of \algg.}
    \label{fig:PropMtd}
    \vspace{-0.1in}
\end{figure*}

\section{\algg}
\label{sec:propMtd}
We present the \alg design for uplink communications in wireless FL, and give some preliminary analysis.

% In this section, we propose our novel design \alg for uplink communication in wireless FL. By incorporating the unique characteristics of orthogonal sequences, \alg enables the BS to directly obtain the estimates of aggregated parameters, which not only reduces receiver complexity but also achieves better privacy protection.

\subsection{Algorithm Design}\label{subsec:alg}

\alg is a four-step protocol detailed as follows.

\mypara{Step 1: Uplink channel estimation.} The BS first schedules \emph{all} participating users to transmit a common pilot $s$ simultaneously. The received signal is
\begin{align*}
    \ybf_s = \sum_{k = 1}^K \abf_k  h_k  s + \nbf_s.
\end{align*}
The BS can utilize $\ybf_s$ and the {\it complete} set of orthogonal sequences $\mathcal{A}=\{ \abf_1, \cdots, \abf_N \}$ to estimate the channel gain coefficients. For the $K$ spreading sequences that are actually adopted by the user\footnote{Without loss of generality, we assume the first $K$ spreading sequences from the set are selected. This assumption is made to ease the notation.}, we have
\begin{equation*}
    \hat h_k  = \frac{\abf_k^T\left[\sum_{i = 1}^K \abf_i h_i s + \nbf_s\right]}{s}  =  h_k + \frac{\abf_k^T\nbf_s}{s}, \forall k = 1,\cdots,K.
\end{equation*}
% \begin{equation*}
% \begin{split}
%     \hat h_k & = \frac{\abf_k^T \ybf_s}{s} =  \frac{\abf_k^T\left[\sum_{k = 1}^K \abf_k h_k s + \nbf_s\right]}{s} \\
%     & = h_k + \frac{\abf_k^T\nbf_s}{s}, \;\;\forall k = 1,\cdots,K.
% \end{split}
% \end{equation*}
For the $N-K$ spreading sequences that are not selected by any user, the BS obtains
\begin{equation*}
    \hat h_k = \frac{\abf_k^T\left[\sum_{i = 1}^K \abf_i h_i s + \nbf_s\right]}{s} = \frac{\abf_k^T\nbf_s}{s}, \forall k = K + 1,\cdots,N.
\end{equation*}
% \begin{equation*}
% \begin{split}
%     \hat h_k & = \frac{\abf_k^T \ybf_s}{s} =  \frac{\abf_k^T\left[\sum_{k = 1}^K \abf_k h_k s + \nbf_s\right]}{s} \\
%     & = \frac{\abf_k^T\nbf_s}{s}, \;\; \forall k = K + 1,\cdots,N.
% \end{split}
% \end{equation*}
% We note, however, that the BS does not know which sequences are used. As a result, it simply obtains $N$ estimates $\{\hat h_1,  \cdots, \hat h_N\}$.
We emphasize that the BS is not able to distinguish these two cases; all it has are $N$ estimates $\{\hat h_1, \cdots, \hat h_N\}$.

\mypara{Step 2: Projector construction.} For simplicity, we assume $s = 1$ in the following discussion. After the channel estimation, the BS constructs the following vector based on \emph{all} of the estimated channel coefficients:
\begin{equation*}
    \vbf = \sum_{k = 1}^N \frac{1}{\hat h_k} \abf_k = \sum_{k = 1}^K\frac{\abf_k}{h_k + \abf_k^T\nbf_s} + \sum_{k = K + 1}^N \frac{\abf_k}{\abf_k^T\nbf_s}.
\end{equation*}
We note that since the BS does not know which $K$ of the total $N$ spreading sequences are adopted by the users, it has to use all of $\{\hat h_1,  \cdots, \hat h_N\}$ to construct the projector $\hbf_s$. This seemingly redundant design actually enables better privacy protection, which will be clear in Section \ref{sec:privacy}.

\mypara{Step 3: UL model transmission.} All users transmit every element of the differentials via $d$ shared time slots:
\begin{equation*}
    \ybf_i = \sum_{k = 1}^K \abf_k h_k x_k^i + \nbf_i, \;\;\forall i = 1,\cdots,d.
\end{equation*}
% The BS can then use $\ybf_i$ to estimate aggregated parameters by the following step.

\mypara{Step 4: Sum model decoding.} The BS applies the following linear projection to estimate each aggregated model differential $x_i, \forall i = 1,\cdots,d$: 
\begin{align*}
%\begin{split}
    \tilde x_i & = \vbf^T\ybf_i = \sum_{k = 1}^N \frac{1}{\hat h_k} \abf_k  \squab{\sum_{k = 1}^K \abf_k h_k x_k^i + \nbf_i}\\
    & = \squab{\sum_{k = 1}^K\frac{\abf_k}{h_k + \abf_k^T\nbf_s} + \sum_{k = K + 1}^N \frac{\abf_k}{\abf_k^T\nbf_s}} \squab{\sum_{k = 1}^K \abf_k h_k x_k^i + \nbf_i}\\
    & =  \sum_{k = 1}^K \frac{h_k}{h_k + \abf_k^T\nbf_s}x_k^i + \sum_{k = 1}^K \frac{\abf_k^T\nbf_i}{ h_k + \abf_k^T\nbf_s}  + \sum_{k = K + 1}^N \frac{\abf_k^T\nbf_i}{\abf^T_k \nbf_s}.%, \\
    %&\;\;\forall i = 1,\cdots,d.
%\end{split}
\end{align*}
After obtaining $\{\tilde x_1, \cdots, \tilde x_d \}$, the BS can compute the new global model following \eqref{eq:diffGlobal} and start the next learning round. The above four-step procedure is illustrated in Fig.~\ref{fig:PropMtd}. %It is easy to find that $\tilde x_i$ is an \emph{unbiased estimator} of the aggregation parameters. 

\subsection{Preliminary Analysis}
In the high signal-to-noise ratio (SNR) regime, where the channel fading effect dominates the noise, we have $\expt[\norm{\abf_k^T\nbf_s}^2]\ll  \expt[\norm{h_k}^2]$. Therefore, we can establish the following approximation for the estimated model in Step $4$:
\begin{equation}\label{eq:approx}
    \tilde x_i \approx   \sum_{k = 1}^K x_k^i  + \sum_{k = K + 1}^N \frac{\abf_k^T\nbf_i}{\abf^T_k \nbf_s}, \;\;\forall i = 1,\cdots,d,
\end{equation}
where $\sum_{k = K + 1}^N \frac{\abf_k^T\nbf_i}{\abf^T_k \nbf_s}$ denotes the dominant noise of the received global model parameters\footnote{Note that this approximation drops the minor noise term, which results in that the DP guarantee in the later discussion is a lower bound of the true DP level of \algg, i.e., we achieve better DP than that computed in this paper.}. The distribution of this post-processing noise is not straightforward, and we next present Lemma \ref{lemma:cauchyRndVar} to establish that the noise term is a \emph{Cauchy random variable}.

\begin{lemma}\label{lemma:cauchyRndVar}
\zixiangbbb{Define $\gamma\triangleq N - K$}. For IID Gaussian random vector $\nbf_i,\nbf_s \sim \mathcal{N}(0, \f{\sigma^2}{L}\Ibf)$, random variable $$X\triangleq \sum_{k = K + 1}^N \frac{\abf_k^T\nbf_i}{\abf^T_k \nbf_s}\sim \ssf{Cauchy}(0,\zixiangbbb{\gamma}),\;\;\forall\abf_k\in \mathcal{A},$$
with probability density function (PDF)
\begin{equation*}
f_X(x)= \frac{1}{\pi}\frac{\zixiangbbb{\gamma}}{x^2 + \zixiangbbb{\gamma}^2},\;\;x\in \mathbb{R}.
\end{equation*}
\end{lemma}

\begin{proof}
We first note that $\abf^T_k\nbf_i$ and $\abf^T_k\nbf_s$ are Gaussian random variables since they are linear combinations of IID Gaussian random variables. Let $$ Y \triangleq \squab{\abf_{K + 1}^T\nbf_i, \abf_{K + 2}^T\nbf_i,\cdots,\abf_{N}^T\nbf_i}^T$$ and $$Z \triangleq \squab{\abf_{K + 1}^T\nbf_s, \abf_{K + 2}^T\nbf_s,\cdots,\abf_{N}^T\nbf_s}^T,$$
% \begin{equation*}
%     Y \triangleq \squab{\abf_{K + 1}^T\nbf_i, \abf_{K + 2}^T\nbf_i,\cdots,\abf_{N}^T\nbf_i}^T,
% \end{equation*}
% and
% \begin{equation*}
%     Z \triangleq \squab{\abf_{K + 1}^T\nbf_s, \abf_{K + 2}^T\nbf_s,\cdots,\abf_{N}^T\nbf_s}^T.
% \end{equation*}
and it is straightforward to verify that $Y$ and $Z$ are IID Gaussian random vectors with distribution $\mathcal{N}(0, \Sigmabf)$, where $\Sigmabf$ is a covariance matrix. According to \cite{pillai2016unexpected}, we have
\begin{equation*}
    \sum_{k = K + 1}^N w_k\frac{Y_k}{Z_k} = \sum_{k = K + 1}^N b_k\frac{\abf_k^T\nbf_i}{\abf^T_k \nbf_s}\sim \ssf{Cauchy}(0,1),
\end{equation*}
as long as $w_k$ is independent of $(Y,Z)$ and $\sum_{k = K + 1}^N w_k= 1$. Letting $b_k = \frac{1}{\zixiangbbb{\gamma}}$ and using the fact that $kX\sim\ssf{Cauchy}(0,|k|)$ if $X\sim\ssf{Cauchy}(0,1)$, we prove Lemma \ref{lemma:cauchyRndVar}.
\end{proof}

Cauchy distribution is known as a ``fat tail" distribution, as the tail of its PDF decreases proportionally with $1/x^2$. Lemma \ref{lemma:cauchyRndVar} suggests that, for a fixed $K$, a larger spreading sequence set will result in a heavier tail in the Cauchy noise. Therefore, we can adjust the size of the spreading sequence set to induce different additive Cauchy noise in \eqref{eq:approx}. We will discuss the effect of Cauchy noise on DP and convergence in Sections \ref{sec:privacy} and \ref{sec:CvgAna}, respectively. 

A few remarks about \alg are now in order.

\begin{remark}[Advantages over the channel inversion method]
Compared with the widely studied channel inversion-based AirComp design, \alg does not require CSIT for uplink communications, which greatly reduces the communication overhead. This is especially attractive for Internet-of-Things (IoT) applications with massive devices. Moreover, in SISO systems, the maximum uplink transmit power of each user in the channel inversion-based methods is usually limited by the worst channel gain. As a result, the received SNR of the global model and the efficiency of power amplifiers (PAs) will significantly decrease, if one of the client channels experiences deep fading \cite{TV:05}. Thanks to the orthogonality of spreading sequences, \alg allows the transmit power to be independent of small-scale fading channel realizations, and thus avoids increasing the dynamic range of the transmit signal, which improves the power efficiency of PAs.
\end{remark}

\begin{remark}[Leverage full channel degrees of freedom]\label{remark:channelgain}
As mentioned before, real modulation cannot exploit full degrees of freedom of the complex channel. To address this limitation, we can borrow the idea of \cite{sery2020tsp}. Specifically, the fading channel between each client and the BS can be written as $$h_k = \tilde h_k e^{j\phi_k}, \;\;\forall k = 1,\cdots,K,$$ where $\tilde h_k\triangleq |h_k|\in \mathbb{R}_+$ and $\phi_k\triangleq\angle h_k\in[0, 2\pi)$ represent channel amplitude (gain) and phase, respectively. For a given $x_{k}^i$, client $k$ transmits symbol $x_{k}^i e^{-j\phi_k}$, where $e^{-j\phi_k}$ is the phase correction term as suggested in \cite{sery2020tsp}. Hence the received signal at the BS can be written as%\footnote{For simplicity, we assume real signals $\{x_{k}^i\}_{i = 1}^d$ are transmitted in this paper. It can be easily extended to complex signals by stacking two real model parameters into a complex signal, so that the full d.o.f. is utilized.}
\begin{equation*}
    \ybf_i = \sum_{k = 1}^K \abf_k h_k e^{-j\phi_k}x_k^i + \nbf_i = \sum_{k = 1}^K \abf_k  \tilde h_k x_k^i + \nbf_i, \forall i = 1,\cdots,d.
\end{equation*}
% \begin{align*}
%     \ybf_i & = \sum_{k = 1}^K \abf_k h_k e^{-j\phi_k}x_k^i + \nbf_i\\
%     & = \sum_{k = 1}^K \abf_k  \tilde h_k x_k^i + \nbf_i \;\;\forall i = 1,\cdots,d.
% \end{align*}
By phase correction, the imaginary part of $h_k$ is projected to the real domain and the full channel gain can be leveraged without sacrificing any other advantages of \algg. We note that same as in \cite{sery2020tsp}, this approach requires each client $k$ to have the channel phase information $\phi_k$, which is weaker than the complete CSIT $h_k$ but stronger than the standard \algg.

Note that we consider the imperfect channel estimation previously in Step $2$ for uplink communication. For simplicity, we assume that each client has perfect partial CSIT (channel phases) obtained from downlink channel estimation here. One reason for this different consideration is that the transmit power at the BS is usually much larger than those at devices, which naturally ensures more accurate downlink channel estimation than uplink. Additionally, even if the phase correction term suffers from estimation error, as long as the error is within $[-\pi/2,\pi/2)$, the proposed design is still valid \cite{sery2020tsp}, only at a cost of some channel gain lost. %As we do not make any assumptions on channel model in this work, the perfect partial CSIT assumption is sufficient.
\end{remark}

%\begin{remark}[Channel estimation accuracy]
%Note that in our designs, we consider the imperfect channel estimation in Step $2$ for uplink communication. On the other hand, we assume that each client has perfect partial CSIT (channel phases) obtained from downlink channel estimation. One reason for this different consideration is that the transmission power at the BS is usually much larger than clients, which naturally ensure more accurate downlink channel estimation than uplink. Moreover, since the BS has no knowledge on specific assignment of spreading sequence, estimating $\zixiangbbb{\gamma}$ channel gain coefficient as $\abf_k^T\nbf_s$ is unavoidable, which directly contributes to the Cauchy noise in \eqref{eq:approx}. Additionally, even if the phase correction term suffers from estimation error, as long as the error is within $[-\pi/2,\pi/2)$, the proposed design is still valid \cite{sery2020tsp}, only at a cost of some channel gain lost. As we do not make any assumptions on channel model in this work, the perfect partial CSIT assumption is sufficient.
%\end{remark}

\begin{remark}[Extensions on NOMA systems]
Another limitation of \alg is that the number of orthogonal spreading sequences is fixed for a given sequence length. To expand the size of set $\mathcal{A}$, the system would need to adopt longer spreading sequences, which consumes more bandwidth. We first note that although there may exist a large number of clients in an FL task, the number of actively participating clients in each learning round is usually relatively small (due to client selection), which implies the bandwidth cost will not be too significant for our design. Second, as an alternative, the system can adopt non-orthogonal spreading sequences to improve the scalability without the cost of extra bandwidth. Applying non-orthogonal spreading sequences is consistent with non-orthogonal multiple access (NOMA) systems, which is an emerging technology for massive machine-type communications (mMTC) applications. For non-orthogonal spreading sequences, the requirement in \eqref{eq:orthogonalSC} becomes
\begin{equation*}
    \expt\squab{\abf_i^T\abf_i} = 1,\;\;\forall i  \text{~~and~~}\expt\squab{\abf_i^T\abf_j} = 0, \;\;\forall i\neq j,
\end{equation*}
which can be achieved by random Gaussian vectors. Note that Lemma \ref{lemma:cauchyRndVar} still holds for non-orthogonal spreading sequences, since it allows an arbitrary covariance matrix of random Gaussian vectors $Y$ and $Z$. Non-orthogonal spreading sequences will introduce inter-symbol interference besides noise when decoding the global model parameters in Step $4$. Therefore, the convergence bound developed in Section \ref{sec:CvgAna} can be regarded as a lower bound for NOMA systems. %For more analysis in the presence of interference, readers can refer to \cite{weiRO2022}.

\end{remark}

\section{Differential Privacy Analysis}\label{sec:privacy}
In this section, we analyze the DP level achieved by \algg. We begin by introducing the basic concepts of DP in FL, and then prove that \alg achieves different levels of DP via the adjustment of the size of spreading sequence set $N$ and the number of involved clients $K$. Both item-level and client-level DP are analyzed. The DP guarantee not only considers multiple sources of randomness in wireless FL, including random mini-batch in SGD, the Cauchy noise, and the random client participation, but also reveals the influence of multiple learning rounds. 
%For ease of analysis, we denote a noise-free global model parameter as $\xbf_t\triangleq \sum_{k = 1}^K \mathsf{x}_t^k$.

\subsection{Preliminaries}
We first introduce the concept of \emph{neighboring datasets}. We say that two datasets $\mathcal{D}$ and $\mathcal{D}^{\prime}$ are neighboring, written as $\mathcal{D}\sim \mathcal{D}^{\prime}$, if they differ in at most one sample.  
Based on this concept, we state the standard definition of $(\epsilon, \delta)$-DP as follows.
\begin{defi}[$(\epsilon, \delta)$-DP \cite{dwork2014algorithmic}]\label{defi:dp}
A randomized algorithm $\mathcal{M}:X^n\rightarrow \mathcal{R}$ provides $(\epsilon, \delta)$-DP with $0 \leq \delta <1$, if for all pairs of neighboring datasets $\mathcal{D}\sim \mathcal{D}^{\prime}$ and all measurable sets of outcomes $S\subseteq \mathcal{R}$, we have
\begin{equation*}
    \prob[\mathcal{M}(\mathcal{D})\in S]\leq e^\epsilon \prob[\mathcal{M}(\mathcal{D}^{\prime})\in S] + \delta.
\end{equation*}
\end{defi}
We say that a randomized algorithm achieves $\epsilon$-DP (also known as \emph{pure} DP) if it satisfies $(\epsilon,\delta)$-DP with $\delta = 0$. The common interpretation of $\delta$  is the ``leakage probability", i.e., $(\epsilon,\delta)$-DP is $\epsilon$-DP ``except with probability $\delta$”. 

\zixiangbb{We next introduce the definition of Rényi DP \cite{mironov2017renyi}. As a generalization of $(\epsilon,\delta)$-DP, it can help us obtain a tighter $(\epsilon,\delta)$-DP bound converted from its composition, which is particularly attractive in FL due to its multiple learning rounds.}
\begin{defi}[$(\alpha,\epsilon)$-Rényi DP \cite{mironov2017renyi}]
A randomized algorithm $\mathcal{M}:X^n\rightarrow \mathcal{R}$ is said to provide $(\alpha,\epsilon)$-Rényi DP, if for any pair of neighboring datasets $\mathcal{D}\sim \mathcal{D}^{\prime}$ it holds that
\begin{equation*}
    D_\alpha(\mathcal{M}(\mathcal{D})||\mathcal{M}(\mathcal{D}^\prime))\leq \epsilon,
\end{equation*}
where
\begin{equation*}
    D_\alpha(\mathcal{M}(\mathcal{D})||\mathcal{M}(\mathcal{D}^\prime))\triangleq \frac{1}{\alpha - 1}\log \int_{-\infty}^{+\infty} P(x)^\alpha Q(x)^{1 - \alpha} dx
\end{equation*}
is the Rényi divergence. Specially, for $\alpha = + \infty$, we have
\begin{equation*}
    D_{\infty}(\mathcal{M}(\mathcal{D})||\mathcal{M}(\mathcal{D}^\prime)) = \sup_{x \in \text{supp} Q}\log \frac{P(x)}{Q(x)}.
\end{equation*}
\end{defi}
In the AirComp FL design, \zixiangbb{decoding the global model from the received signal can be regarded as a randomized mechanism on \zixiangr{$\mathcal{D} = \bigcup_{k = 1}^M\mathcal{D}_k$}, where $\mathcal{D}$ is the union of the local datasets of all \zixiangr{$M$ clients throughout the whole FL task}. We denote this randomized mechanism as
\begin{equation}\label{eq:CauchyMechanism}
    \mathcal{M}(\mathcal{D})\triangleq \tilde{\xbf}_t = g (\mathcal{D}) + \nbf_t,
\end{equation}
where $g(\mathcal{D}) = \xbf_t$ is the noise-free summation of model differentials \zixiangr{at the $t$-th learning round}, and $\nbf_t$ is a random noise following a certain distribution. The randomness of the mechanism comes from the \dprev{random client participation,} random mini-batch SGD, and the random noise.}

To determine the DP level, we next define the \emph{global sensitivity function} for the operator $g(\cdot)$ as
\begin{equation}\label{eq:GS}
    GS_{g} = \max_{\mathcal{D},\mathcal{D}^{\prime}}\norm{g(\mathcal{D}) - g(\mathcal{D}^\prime)}_2,
\end{equation}
where $\mathcal{D}^\prime$ is a neighboring dataset of $\mathcal{D}$.  As mentioned in Section \ref{sec:model}, for uplink communications, we ensure that $\norm{\mathsf{x}_t^k}_2\leq C$. Therefore, we have 
\begin{align}\label{eq:GSbound}
     GS_{g} & = \max_{\mathcal{D},\mathcal{D}^{\prime}}\norm{g(\mathcal{D}) - g(\mathcal{D}^\prime)}_2 = \max_{\mathsf{x}_t^k,{\mathsf{x}_t^k}^{\prime}}\norm{\mathsf{x}_t^k - {\mathsf{x}_t^k}^\prime}_2  \leq \max_{\mathsf{x}_t^k,{\mathsf{x}_t^k}^{\prime}}\squab{\norm{\mathsf{x}_t^k}_2 + \norm{{\mathsf{x}_t^k}^\prime}_2}\leq 2C,
\end{align}
where ${\mathsf{x}_t^k}^{\prime}$ denotes the noise-free model differential from client $k^{\prime}$ whose local dataset is swapped in a random data sample. 

\subsection{Item-level Differential Privacy}
% We now derive the Rényi DP of \algg. It is worth mentioning that our analyses consider the DP not only from the Cauchy noise, but also from the random mini-batch SGD operation in the local updates of each learning round. 

The Rényi DP guarantee of \alg in a single learning round is given in Theorem~\ref{theo:RDP}. 

%The relationship between the Rényi divergence with $\alpha = \infty$ and $\epsilon$-differential privacy is immediate. A randomized mechanism is $\epsilon$-differentially private if and only if its distribution over any two adjacent inputs $D$ and $D^\prime$ satisfies
%\begin{equation*}
   % D_{\infty}(\mathcal{M}(\mathcal{D})||\mathcal{M}(\mathcal{D}^\prime))\leq \epsilon.
%\end{equation*}
%Therefore, Theorem \ref{theo:DP} naturally leads to
%\begin{equation*}
   % D_{\infty}(\mathcal{M}(\mathcal{D})||\mathcal{M}(\mathcal{D}^\prime))\leq \frac{4C}{\zixiangbbb{\gamma}},
%\end{equation*} which directly ensures the general Rényi different privacy guarantee of the proposed method.
%We now establish the Rényi DP guarantee of the proposed design.
\begin{theorem}\label{theo:RDP}
    Assume a spreading sequence set $\mathcal{A}$ containing $N$ unique sequences and $M$ total clients involved in an FL task. In each learning round, \dprev{$K < \min\{N, M\}$} clients are independently and uniformly randomly selected to participate in FL. With local dataset $\mathcal{D}_k$ of size $D$ and mini-batch size $d_{\text{batch}} \triangleq |\xi| < D$, \alg provides $(\alpha,\epsilon)$-Rényi DP for the global model, where
    \dprev{\begin{equation}\label{eq:rdp}
        \epsilon = \f{1}{2}\alpha\log^2\left(1 + \zixiangbbb{\frac{qp}{1+qp}\frac{ 2C\sqrt{C^2 + \gamma^2} + 2C^2 }{\gamma^2}}\right)=O\left(\f{\alpha q^2p^2 }{\zixiangbbb{\gamma}^2}\right),
    \end{equation}
    with $q\triangleq\frac{d_{\text{batch}}}{D+1-d_{\text{batch}}}$, $p\triangleq \frac{K}{M}$, and $\gamma = N - K$.}
\end{theorem}
\begin{proof}
See Appendix \ref{app:proofRDP}.
\end{proof}

Based on the composition rule of Rényi DP, we next establish the $(\epsilon, \delta)$-DP guarantee for the overall uplink communications in an FL task of $T$ rounds.

\begin{theorem}\label{theo:compositeRDP}
Consider a wireless FL task with $T$ learning rounds, a spreading sequence set $\mathcal{A}$ containing $N$ unique sequences, and \dprev{$K < \min\{N, M\}$} clients are independently and uniformly randomly selected in each round. With local dataset $\mathcal{D}_k$ of size $D$ and mini-batch size $d_{\text{batch}} < D$, \alg provides $(\epsilon^\prime,\delta)$-DP for the entire FL task, where
    \begin{align}\label{eq:compositerdp}
        \epsilon^\prime & = \sqrt{2T\log(1/\delta)}\log\left(1 + \zixiangbbb{\dprev{\frac{qp}{1+qp}}\frac{ 2C\sqrt{C^2 + \gamma^2} + 2C^2 }{\gamma^2}}\right) \nonumber \\
        & + \f{1}{2}T\log^2\left(1 + \zixiangbbb{\frac{qp}{1+qp}\frac{ 2C\sqrt{C^2 + \gamma^2} + 2C^2 }{\gamma^2}}\right) \sim \dprev{\tilde{O}\left(\f{\sqrt{T}qp}{\zixiangbbb{\gamma}}\right),
        }
    \end{align}
\dprev{with $q\triangleq\frac{d_{\text{batch}}}{D+1-d_{\text{batch}}}$, $p\triangleq \frac{K}{M}$, and $\gamma = N - K$.} 
    
\end{theorem}
\begin{proof}
See Appendix \ref{app:ProofCompositeRDP}.
\end{proof}

\begin{figure}[t]
    \centering
    \subfigure[$q = 0.05, \zixiangbbb{\gamma} = 5$]{\includegraphics[width = 0.32\linewidth]{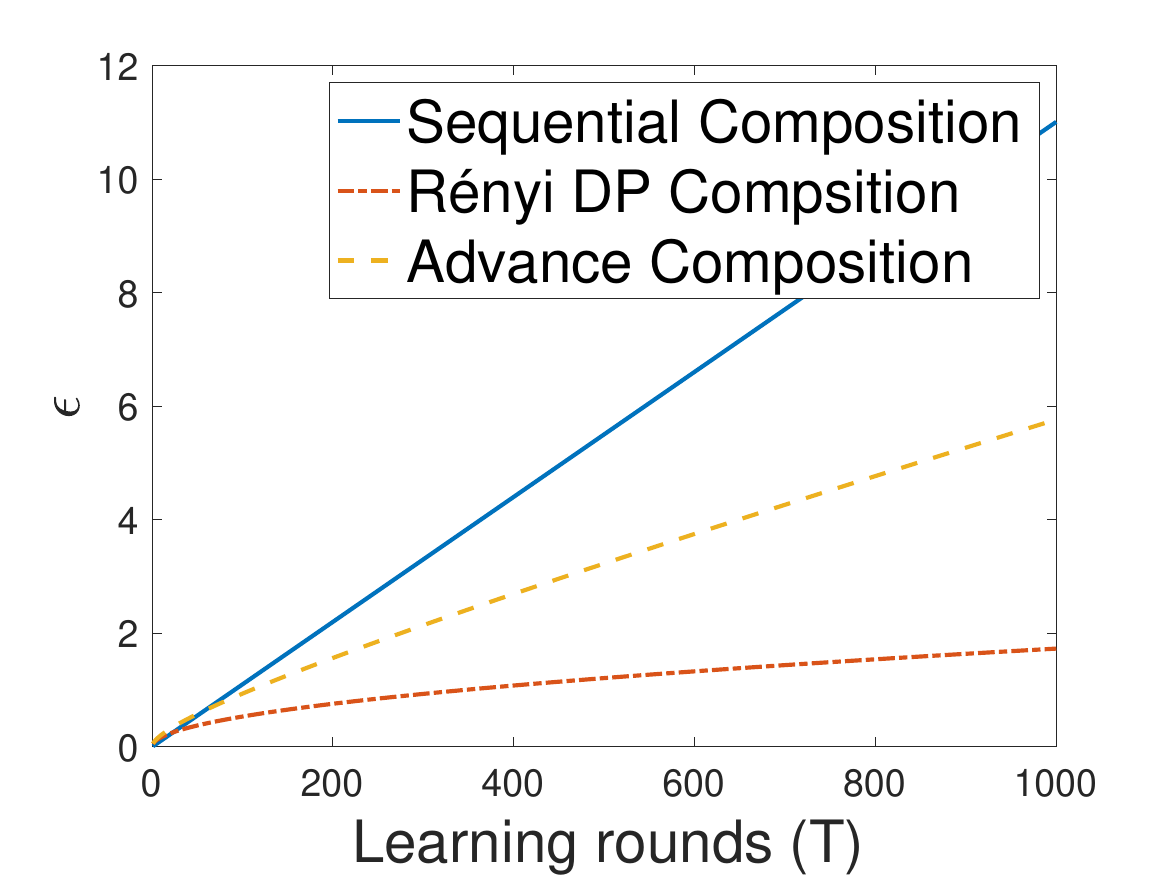}}
    \subfigure[$q = 0.01, \zixiangbbb{\gamma} = 5$]{\includegraphics[width = 0.32\linewidth]{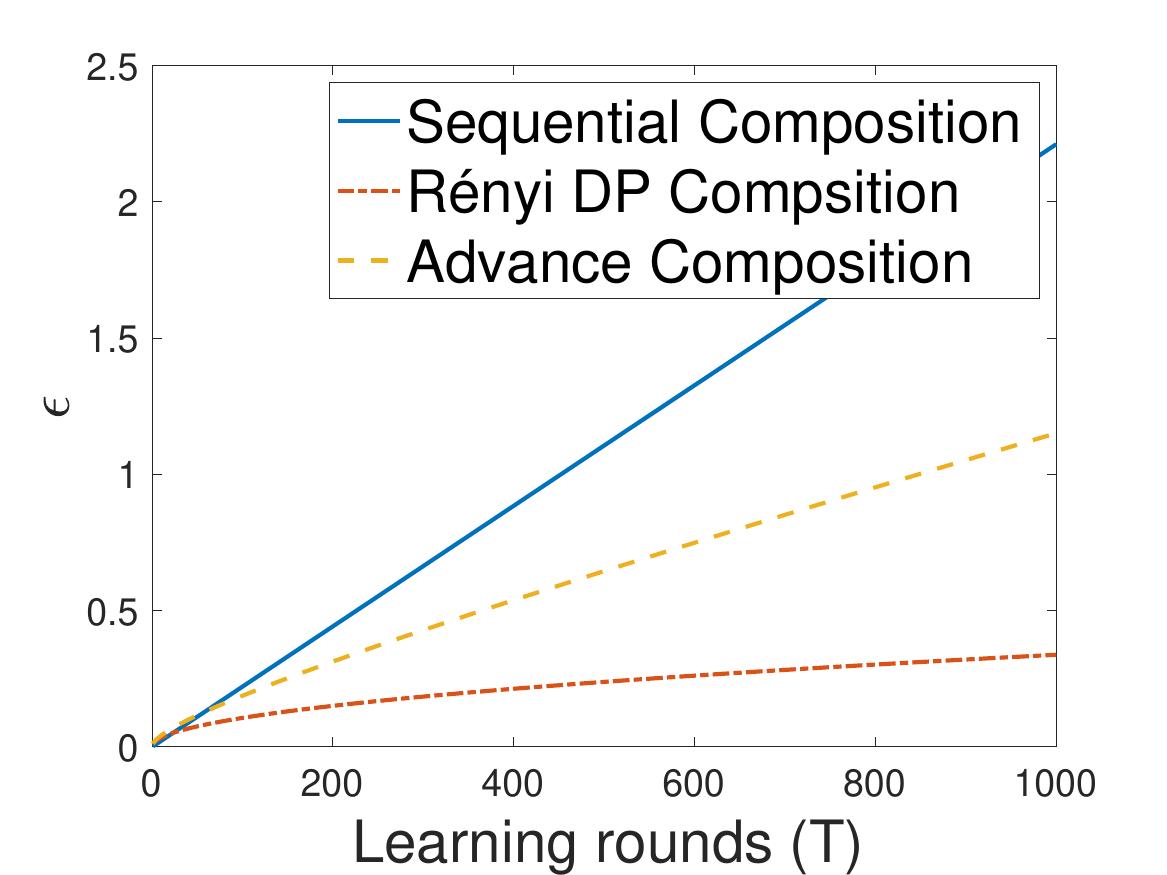}}
    \subfigure[$q = 0.01, \zixiangbbb{\gamma} = 10$]{\includegraphics[width = 0.32\linewidth]{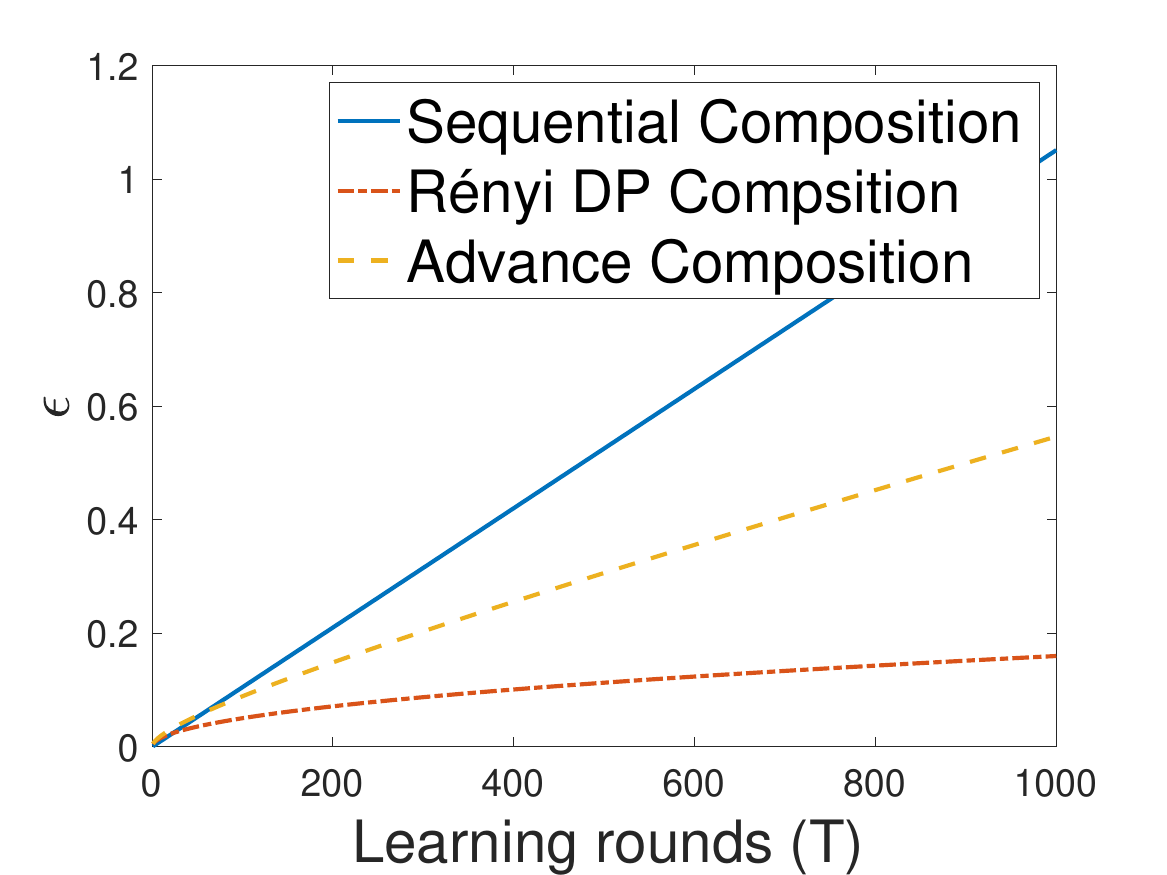}}
    \caption{Comparison of sequential composition, Rényi DP composition, and advanced composition with different sizes of mini-batch SGD and orthogonal sequence set.}
    \label{fig:performanceCompare}
\end{figure}
Theorems \ref{theo:RDP} and \ref{theo:compositeRDP} reveal that, for a given FL configuration, i.e., fixed number of selected clients $K$, mini-batch size ratio $q$, \dprev{participation ratio $p$}, and the number of learning rounds $T$, the expansion of spreading sequence set would achieve a higher level of DP per learning round and for the whole learning task, respectively. In particular, $\epsilon \propto \f{1}{\zixiangbbb{\gamma}^2}$ in Rényi DP for each learning round, and $\epsilon^\prime \propto \f{1}{\zixiangbbb{\gamma}}$ in $(\epsilon,\delta)$-DP for the whole learning task. Since the BS (adversary) has no knowledge of which particular $K$ out of the total $N$ spreading sequences the clients have chosen, increasing the number of spreading sequences results in a heavier tail of the post-processing Cauchy noise, which achieves better privacy protection. 

Guided by Theorem \ref{theo:compositeRDP}, we can adjust $N$ to meet the privacy requirement of a practical FL task. Note that larger noise (better privacy protection) will affect the convergence rate of FL, and we will discuss this impact in detail in Section \ref{sec:CvgAna}. 

\begin{remark}[Choice of DP metrics]
  Note that we use $(\epsilon,\delta)$-DP to evaluate the overall DP guarantee of the FL task. The reason why we choose to utilize Rényi DP to analyze the DP guarantee per learning round is the same as \cite{mironov2019r}: we can obtain a tighter $(\epsilon,\delta)$-DP guarantee from the composition of multiple Rényi-DP mechanisms. This advantage is empirically shown in Fig.~\ref{fig:performanceCompare}, in which we compare the overall DP guarantees v.s. learning rounds of three different composition methods: i) vanilla sequential composition of $(\epsilon,\delta)$-DP; ii) advanced composition of $(\epsilon,\delta)$-DP; iii) $(\epsilon,\delta)$-DP converted from composition of Rényi DP. It clearly demonstrates that under different system configurations, Theorem \ref{theo:compositeRDP} unanimously provides the tightest $(\epsilon,\delta)$-DP guarantee, \zixiangbb{which is consistent with the theoretical and experimental results in \cite{mironov2017renyi}.} 
\end{remark}

\subsection{Client-level Differential Privacy}
So far, we have focused on the standard item-level DP that protects a single data sample of a certain local dataset. For FL, another DP concept called \emph{client-level DP} (also known as user-level DP) is also important. %Client-level DP protects data from a single client, which requires that the server cannot identify the participation of one client by observing the received signals.
The definition of client-level DP follows similarly from Definition \ref{defi:dp}, with a slight change that the neighboring dataset pair $\mathcal{D}\sim\mathcal{D}^\prime$ differs in at most \emph{all data samples of one single client}. As a result, client-level DP protects privacy when the entire data from a certain client is swapped. Intuitively, this guarantees that the participation of a client cannot be inferred by observing the received signals. 

From the previous discussion of global sensitivity, the output $\textsf{x}_t^k$ is always bounded even though the entire dataset of a certain client changes. Therefore, our method intuitively satisfies the client-level DP, and we formally establish the following guarantee.

\begin{theorem}\label{theo:UL-RDP}
    Given a spreading sequence set $\mathcal{A}$ containing $N$ unique sequences, \alg provides $(\alpha,\epsilon)$-client level Rényi DP, \dprev{when $K$ clients are independently and uniformly randomly selected from the total $M$ clients in each learning round}, where
    \begin{equation}\label{eq:ul-rdp}
        \epsilon = \f{1}{2}\alpha\log^2\left(1 + \zixiangbbb{\dprev{p}\frac{ 2C\sqrt{C^2 + \gamma^2} + 2C^2 }{\gamma^2}}\right)=O\left(\f{\alpha\dprev{p^2} }{\zixiangbbb{\gamma}^2}\right).
    \end{equation}
    Moreover, for total $T$ learning rounds, \alg provides $(\epsilon^{\prime\prime},\delta)$-client level DP for the entire FL task, where
    \begin{align}\label{eq:composite-ulrdp}
        \epsilon^{\prime\prime} & = \sqrt{2T\log(1/\delta)}\log\left(1 + \zixiangbbb{\dprev{p}\frac{ 2C\sqrt{C^2 + \gamma^2} + 2C^2 }{\gamma^2}}\right) \nonumber \\
        & + \f{1}{2}T\log^2\left(1 + \zixiangbbb{\dprev{p}\frac{ 2C\sqrt{C^2 + \gamma^2} + 2C^2 }{\gamma^2}}\right) \sim \tilde{O}\left(\f{\sqrt{T}\dprev{p}}{\zixiangbbb{\gamma}}\right).
    \end{align}
\end{theorem}
\begin{proof}
    See Appendix \ref{app:proofULRDP}.
\end{proof}

Theorem \ref{theo:UL-RDP} demonstrates that \alg provides the client-level DP guarantee in a manner that is similar to the item-level DP. However, since the entire local dataset of a certain client would be swapped, we cannot take advantage of the SGD randomness in the analysis of client-level DP. Therefore, $q$ disappears from the client-level DP guarantee.

\subsection{Spreading Sequence Assignment Mechanism}
\label{subsec:AgnmntMchsm}
As we have discussed, the DP guarantee of \alg comes from the fact that the assignment of spreading sequences remains unknown to the BS (adversary). In traditional CDMA systems, the spreading sequence is assigned to each device by the BS. This mechanism becomes invalid in our setting, since we need to make sure that the BS only has knowledge of the spreading sequence set $\mathcal{A}$, not the individual assignment. In practice, we can resort to a trusted third-party to handle the assignment of orthogonal sequences. The design of such a mechanism belongs to the field of \emph{secure multi-party computation (MPC)} \cite{goldreich1998secure,du2001secure} and is out of the scope of this paper. 

In the following, we provide a preliminary reference design based on a random permutation algorithm as illustrated in Fig.~\ref{fig:AgnmntMchsm}. The proposed design allows each user to autonomously choose a unique spreading code without collision. To better explain the mechanism, we first define a \emph{fixed} matrix based on $\mathcal{A}=\{\abf_k, k = 1, \cdots, N \}: \Abf = \squab{\abf_1, \cdots, \abf_N}\in\mathbb{R}^{L\times N}.$ 
We assume that every device that will be involved in the FL task shares a common confidential \texttt{Key}. This key is \emph{a priori} knowledge to clients, yet confidential to the BS (adversary). \zixiangbb{This assumption can be achieved via standard cryptography approaches, e.g., the key-exchange protocol~\cite{diffie2022new}.} When the BS schedules clients to participate in the current learning round, an \texttt{index} $\in \{1,\cdots,K\}$ will be assigned to each client. After that, every client leverages the confidential \texttt{Key} and the current system time $\texttt{SystemTime()}$ to generate a random seed $\texttt{RandSeed()}$. Since the system has been synchronized, each client will obtain the same random seed, and generate a (common) random permutation matrix $\Pbf_{\pi}\in\mathbb{R}^{N\times N}$ based on that. Followed by a random column permutation
$\tilde \Abf = \Abf\Pbf_\pi,$ 
each client can use the $\texttt{index}$-th column $\tilde \Abf$\texttt{[:,index]} as its spreading sequence for the current learning round. Note that as long as the confidential key is not leaked, the BS (adversary) will not know which $K$ of the total $N$ spreading sequences are adopted in the current learning round.

% \zixiangbbb{
\begin{remark}
We emphasize that the established DP guarantee is for the receiver processing proposed in Section \ref{subsec:alg}. There may exist more sophisticated receiver algorithms, which could conceivably attempt to infer more information about the usage of the orthogonal sequences per learning round from the received signal. However, since we do not make any assumptions on the channel distribution, such approach would be difficult in general. Besides, it is impossible to always accurately identify all used sequences due to the channel noise. Therefore, even under this scenario, the proposed framework still provides privacy, although the DP guarantee dependency may decrease from $O(1/\zixiangbbb{\gamma})$ to $O(1/(\alpha \gamma))$, where $0<\alpha<1$.
\end{remark}
% }

\begin{figure}
    \centering
    \includegraphics[width = \linewidth]{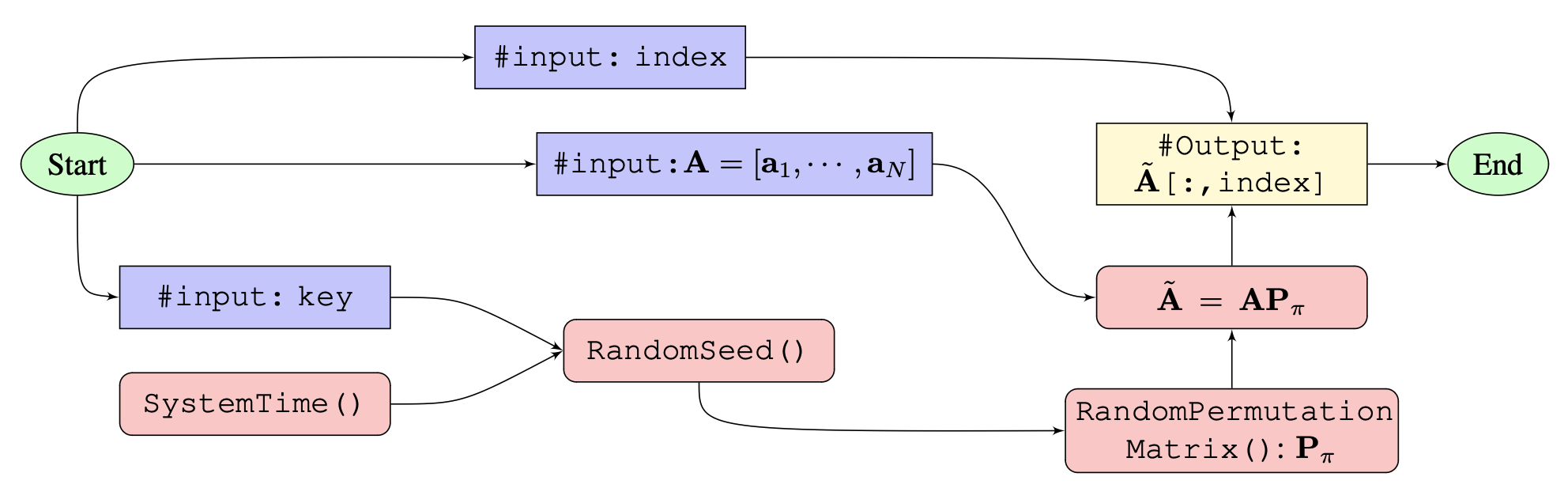}
    \caption{Flowchart of the proposed spreading sequence assignment reference design for any given client.}
    \label{fig:AgnmntMchsm}
\end{figure}

\section{Convergence analysis}
\label{sec:CvgAna}
We analyze the ML model convergence behavior of \alg in this section. We first make the following standard assumptions that are commonly adopted in the convergence analysis of \textsc{FedAvg} and its variants; see \cite{li2019convergence,jiang2018nips,stich2018local,Zheng2020jsac}.  %In particular, Assumption \ref{as:F}-2) indicates that we focus on strongly convex $F_k(\cdot)$, which represents a category of loss functions that are widely studied in the literature.  
%In particular, Assumption \ref{as:F}-1) indicates that the gradient of $f_i$ is Lipschitz continuous. The strongly convex loss function assumption in \ref{as:F}-2) is a category of loss functions that are widely studied in the literature. [list some reference.] [Explain Assumption \ref{as:F}-3).]
% Assumption \ref{as:F}-4) implies that the variance of stochastic gradients is uniformly bounded \cite{stich2018local}.
\begin{assumption}\label{as:1}
\textbf{$L$-smooth:} $\forall~\vect{v}$ and $\vect{w}$, $\norm{f_k(\vbf)-f_k(\wbf)}\leq L \norm{\vbf-\wbf}$;
\end{assumption}
\begin{assumption}\label{as:2}
\textbf{$\mu$-strongly convex:} $\forall~\vect{v}$ and $\vect{w}$, $\left<f_k(\vbf)-f_k(\wbf), \vbf-\wbf\right>\geq \mu \norm{\vbf-\wbf}^2$;
\end{assumption}
\begin{assumption}\label{as:3}
\textbf{Unbiased SGD:} $\forall k \in [M]$, $\expt[\nabla\tilde{f}_k(\wbf)] = \nabla f_k(\wbf)$;
\end{assumption}
\begin{assumption}\label{as:4}
\textbf{Uniformly bounded gradient:} $\forall k \in [M]$, $\expt\norm{\nabla\tilde{f}_k(\wbf)}^2 \leq H^2$ for all mini-batch data.
\end{assumption}

The main challenge, and hence the novelty of our analysis, lies in the Cauchy distributed post-processing noise. We note that a Cauchy distribution has uncertain (infinity) variance. 
% We notice that a Cauchy distribution has uncertain (infinity) variance, which brings a significant challenge into the convergence of \algg, as the noise variance cannot be bounded. 
To address this issue, the BS applies a \emph{truncation} operation in the interval $[-B, B]$ on the decoded global parameters in \eqref{eq:approx}, before de-normalization:
% Therefore, after the BS decode the global parameters in \eqref{eq:approx}, it should truncate the decoding result in the interval $[-B, B]$, 
\begin{equation}\label{eq:truncatedSig}
    \tilde x_i \approx  \max\left( \min \left(\sum_{k = 1}^K x_k^i  + \sum_{k = K + 1}^N \frac{\abf_k^T\nbf_i}{\abf^T_k \nbf_s}, B \right), -B\right)
    % \tilde x_i \approx   \text{truncation}\squab{\sum_{k = 1}^K x_k^i  + \sum_{k = K + 1}^N \frac{\abf_k^T\nbf_i}{\abf^T_k \nbf_s}}_{-B}^{B}
\end{equation}
where $B\gg C$ and $C$ is a normalization parameter defined in Section~\ref{sec:commmodel}. Note that the truncation operation is universal (albeit sometimes implicit) in almost all practical systems, since the signal values in the processing units are always finite. \zixiangr{As a post-processing procedure, the truncation operation has no impact on the DP guarantee of \algg. Unfortunately, it will introduce a bias in the estimate $\tilde x_i$, which impacts convergence. In particular, the equivalent noise after the limiting operation will be a \emph{truncated Cauchy distribution} with support $[-B -\sum_{k = 1}^K x_k^i, B -\sum_{k = 1}^K x_k^i]$}. \zixiangr{Fortunately, as long as we ensure $B\gg C$ which requires a very mild truncation and thus is easy to satisfy, this operation will only have very limited impact on the received signal, hence does not significantly harm the convergence performance. We also note that biased gradients are very common in various DP-guaranteed SGD methods \cite{du2021dynamic,song2021evading,chen2020understanding}. Moreover, after de-normalization, the biased term is effectively diminishing. Therefore, as we show in Theorem \ref{thm.Conv}, \alg can still maintain a good convergence performance despite of the biased noise, which will also be numerically corroborated in Section \ref{sec:exp}.}

\zixiangr{
\begin{theorem}
\label{thm.Conv} 
With Assumptions 1-4 and $\mu > 2\sqrt{dD(\gamma)}EH/K$, for some $r\geq 0$, if we set the learning rate as $\eta_t = \frac{2}{\mu^\prime(t+r)}$, a wireless system implementing \alg achieves
\begin{equation*}
\begin{split}
     & \expt[f(\wbf_t)] - f^* \leq \frac{L}{(t + r)}
      \left[\frac{4G}{{\mu^\prime}^2} + (1 + r)\norm{\wbf_0 - \wbf^*}^2\right],
\end{split}
\end{equation*}
for any $t\geq 1$, where $\mu^\prime \triangleq \mu - 2\sqrt{dD(\gamma)}EH/K$,
\begin{equation*}\small
 \begin{split}
     G \triangleq & \left(1 + \frac{2\sqrt{dD(\gamma)}}{K}EH\eta_1 \right)  \left(\sum_{k=1}^M \frac{H_k^2}{M^2} + 6L \Gamma + 8(E-1)^2 H^2 + \frac{M-K}{M-1} \frac{4}{K}  E^2 H^2 \right)  + {\frac{4dD(\gamma)}{K^2}  E^2 H^2}.
    \end{split}
\end{equation*}
and
\begin{equation*}
    D(\gamma) = \frac{{\gamma}^2}{C^2\arctan\left(\frac{B+C}{{\gamma}}\right)}\squab{\frac{B}{{\gamma}} - \arctan\left(\frac{B+C}{{\gamma}}\right)}.
\end{equation*}
\end{theorem}
\begin{proof}
See Appendix~\ref{App:thm.Conv}.
\end{proof}}
Theorem \ref{thm.Conv} demonstrates that \alg preserves the $O(1/T)$ convergence rate of SGD for strongly convex loss functions (compared with the noise-free FL convergence). For a fixed initial point, there are multiple factors in the constant $G$ that affect the convergence rate of FL. In particular, $\sum_{k=1}^K \frac{H_k^2}{K^2}$ reveals the \emph{variance reduction} effect of SGD by involving more clients, and $6L \Gamma$ and $8(E-1)^2 H^2$ capture the influence of non-IID dataset and the number of local epochs, respectively. We note that term $D(\gamma)$ in $G$ demonstrates the impact from Cauchy noise, i.e. level of privacy protection. We also note that $D(\gamma)$ is increasing as $\gamma$ becomes larger. It implies that a higher level of privacy protection will decrease the speed of convergence as constant $D(\gamma)$ becomes larger. Therefore, Theorem \ref{thm.Conv} establishes a tradeoff between privacy protection and convergence rate, which can guide the practical system design.

\begin{figure*}[htb]
    \centering
    \subfigure[IID MNIST]{\includegraphics[width = 0.45\linewidth]{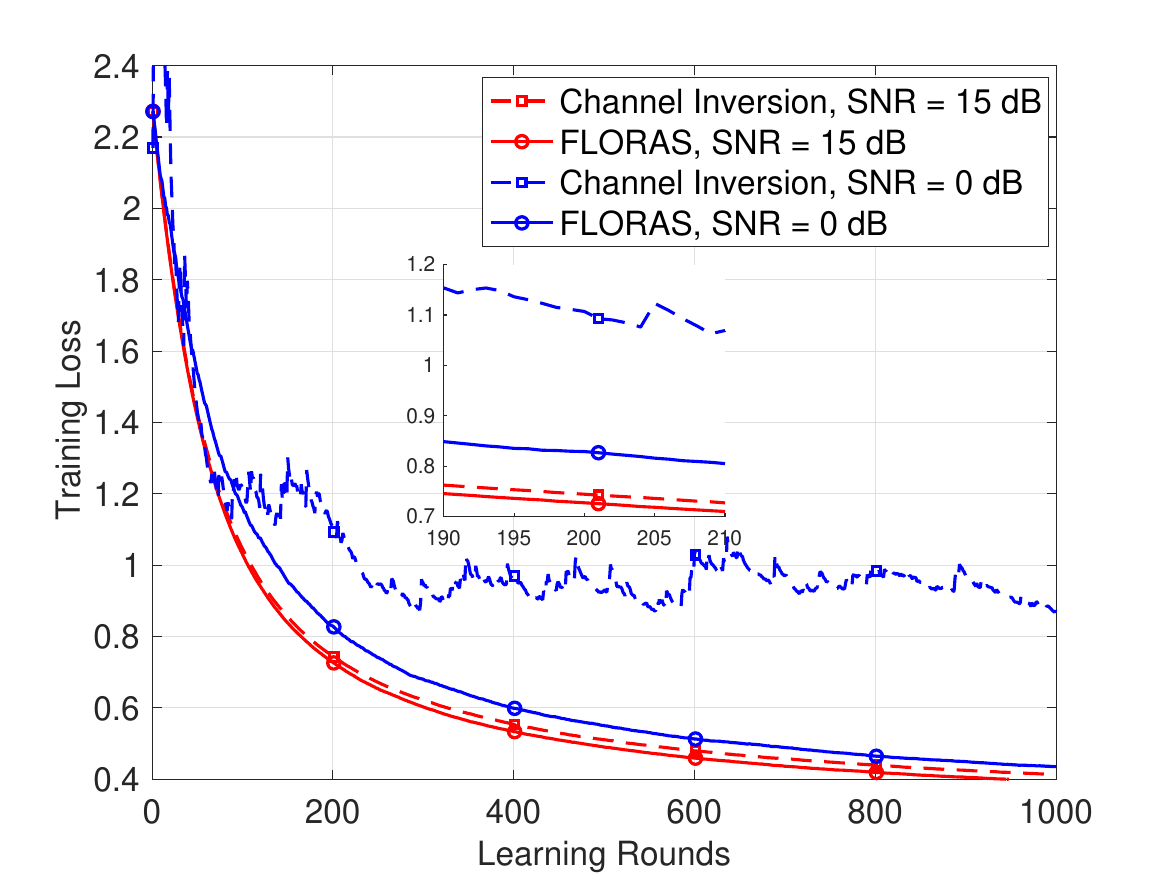}}
    \subfigure[IID MNIST]{\includegraphics[width = 0.45\linewidth]{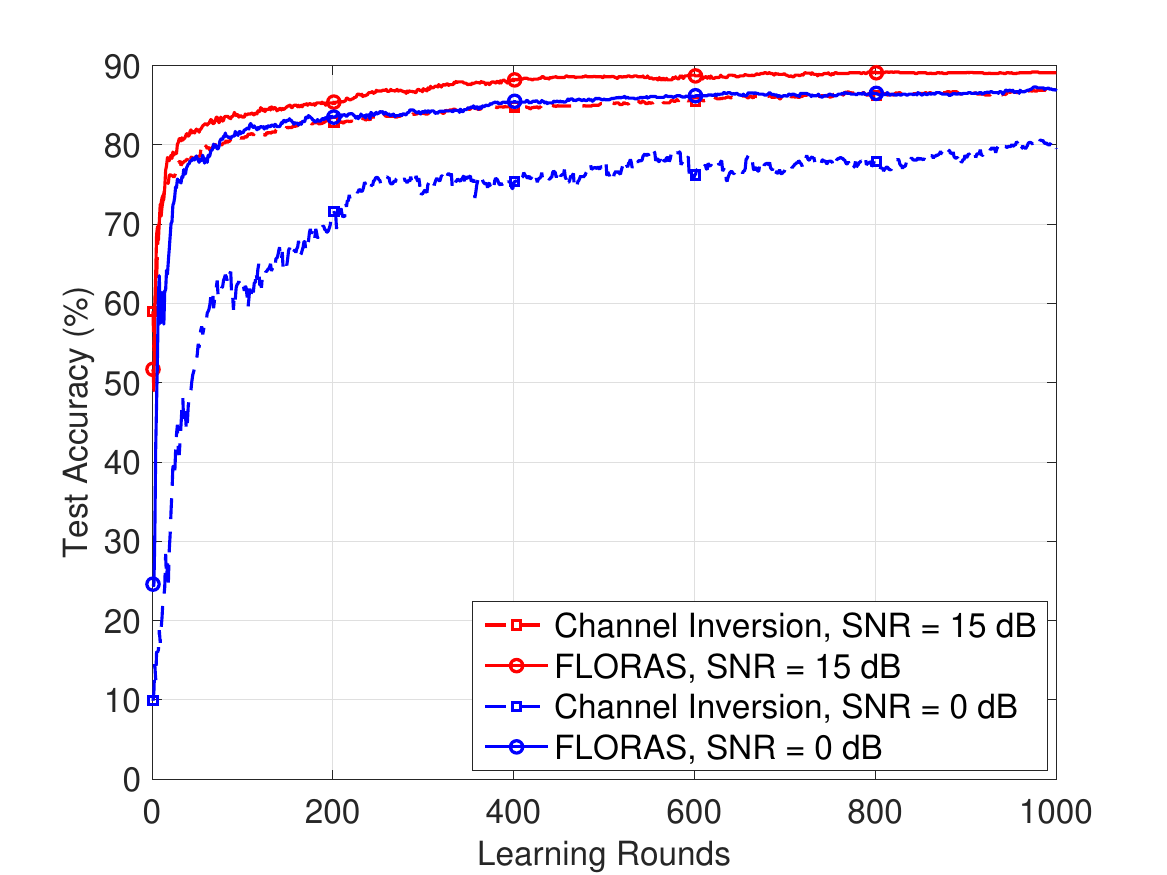}}
    \subfigure[non-IID MNIST]{\includegraphics[width = 0.45\linewidth]{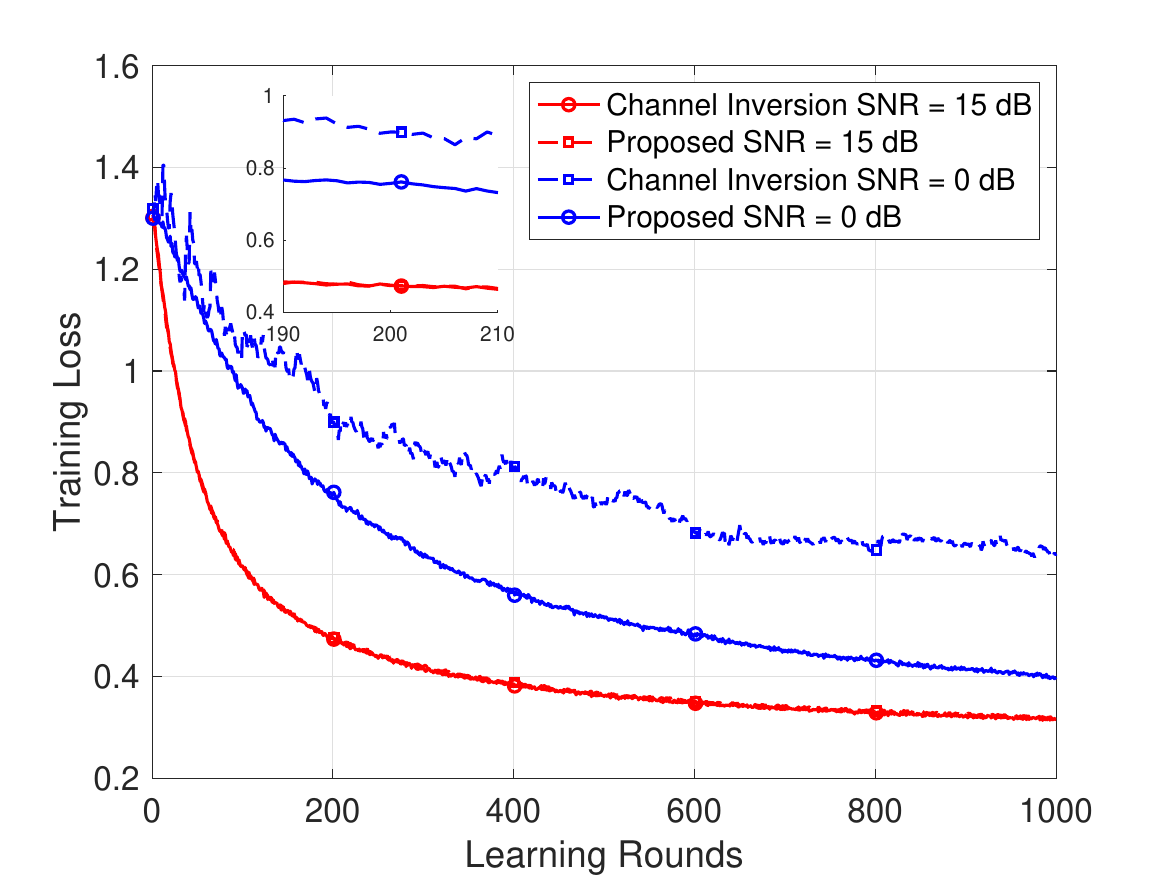}}
    \subfigure[non-IID MNIST]{\includegraphics[width = 0.45\linewidth]{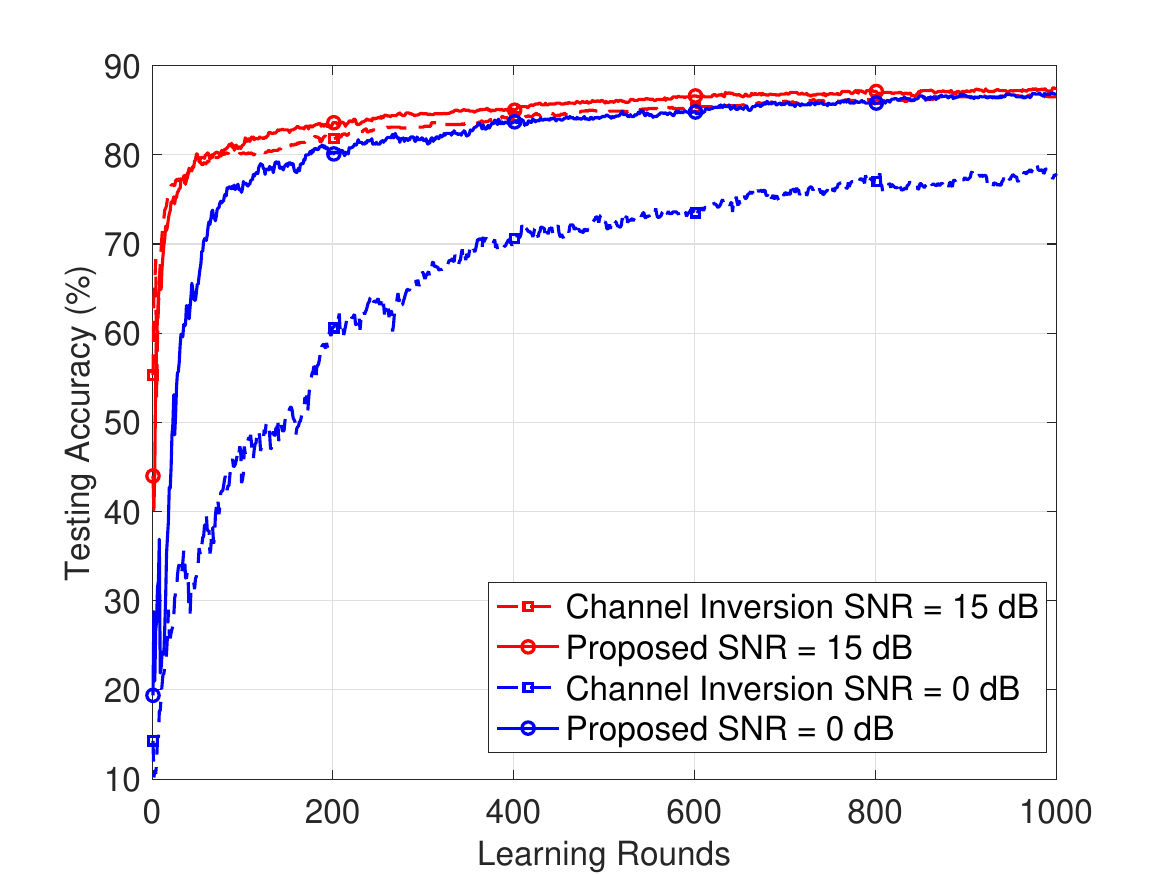}}
    \caption{Training loss/test accuracy versus learning rounds of \alg and channel inversion  for MNIST dataset under SNR $= 0$ and $15$ dB.}
    \label{fig:efficency}
\end{figure*}

\section{Experiments}\label{sec:exp}
In this section, we evaluate the performance of \alg through numerical experiments. We first compare the learning performance of \alg with the widely-investigated channel inversion AirComp method. Then, we evaluate the effect on the ML model convergence rate of various DP levels. In particular, we corroborate the theoretical results via real-world FL tasks on the MNIST dataset. The experiments demonstrate that \alg achieves superior performances under various SNRs and other system configurations. 

Details on the setup of experiments are as follows. All convergence curves are the average of five independent Monte Carlo trials. 
% \begin{enumerate}[leftmargin=*]\itemsep=0pt
        %\item 
        The MNIST dataset contains multiple handwritten digit figures of $20 \times 20$ pixels. The training set contains $4,000$ examples and is evenly distributed over $K = 20$ clients. For the {IID} case, the data is shuffled and randomly assigned to each client; for the {non-IID} case, the data is sorted by labels and each client is randomly assigned with the data of one label. The test set size is $1,000$. We evaluate \alg and validate the theoretical results on a multinomial logistic regression task. Specifically, let $f(\wbf;x_i)$ denote the prediction model with the parameter $\wbf = (\Wbf,\bbf)$ and the form $f(\wbf;x_i) =\texttt{softmax}(\Wbf x_i + \bbf)$. The loss function is given by
      $\texttt{loss}(\wbf) =  \f{1}{D}\sum_{i = 1}^{D}\texttt{CrossEntropy}(f(\wbf;x_i),\ybf_i) + \lambda\norm{\wbf}^2$.
        We adopt the regularization parameter $\lambda = 0.01$ in the experiments. 
        %\item \textbf{CIFAR-10.} We set $N=K=10$ for the full clients participation case while $N=100$ and $K=10$ for the partial clients participation case. We train a CNN model with two $5 \times 5$ convolution layers (both with 64 channels), two fully connected layers (384 and 192 units respectively) with $\ssf{ReLU}$ activation and a final output layer with softmax. The two convolution layers are both followed by $2 \times 2$ max pooling and a local response norm layer. The training parameters are: (a) {IID}: $BS=50$, $E=5$, learning rate initially sets to $\eta=0.15$ and decays every 10 rounds with rate 0.99; (b) {non-IID}: $BS=100$, $E=1$, $\eta=0.1$ and decay every round with rate 0.992.
% \end{enumerate}

\subsection{Communication Efficiency}
We first evaluate the performance of \alg compared with the channel inversion method. We assume IID Raleigh block fading channel $h_k \sim \mathcal{CN}(0, 1), \forall k=1, \cdots, K$. For channel inversion, we adopt a user admission threshold $0.01$ for the fading channel gain to avoid deep fading. The following parameters are used for training: local batch size $50$, the number of local epochs $E=1$, and learning rate $\eta = 0.005$ and $\eta = 0.001$ for the IID and non-IID case, respectively. 

Fig.~\ref{fig:efficency}-(a) and Fig.~\ref{fig:efficency}-(b) illustrate the training loss and test accuracy performance versus learning round of \alg and channel inversion in high (red line) and low (blue) SNR regimes, respectively. For high SNR ($15$ dB), \alg and channel inversion have similar performances. Although transmitter fading channel cancellation in channel inversion limits the maximum transmission power of each user, its performance does not deteriorate, since noise is not the dominant factor of the convergence. However, we note that, unlike \algg, channel inversion requires full CSIT at each client, which not only consumes larger communication overhead, but also increases the dynamic range of the signal, bringing higher hardware cost. The advantages of \alg become conspicuous in the low SNR regime ($0$ dB), in which noise becomes the dominant factor of the convergence rate. \alg allows all participated clients to make full use of transmit power and achieves significantly better performance. As shown in Fig.~\ref{fig:efficency}-(b), \alg achieves about $7.5 \%$ higher test accuracy compared with channel inversion at SNR = $0$ dB. This phenomenon becomes more notable in the non-IID dataset as shown in Fig.~\ref{fig:efficency}-(c) and Fig.~\ref{fig:efficency}-(d), where the performance gap of test accuracy between \alg and channel inversion further grows to $10. 2\%$.

\begin{figure*}[htb]
    \centering
    \subfigure[IID MNIST]{\includegraphics[width = 0.45\linewidth]{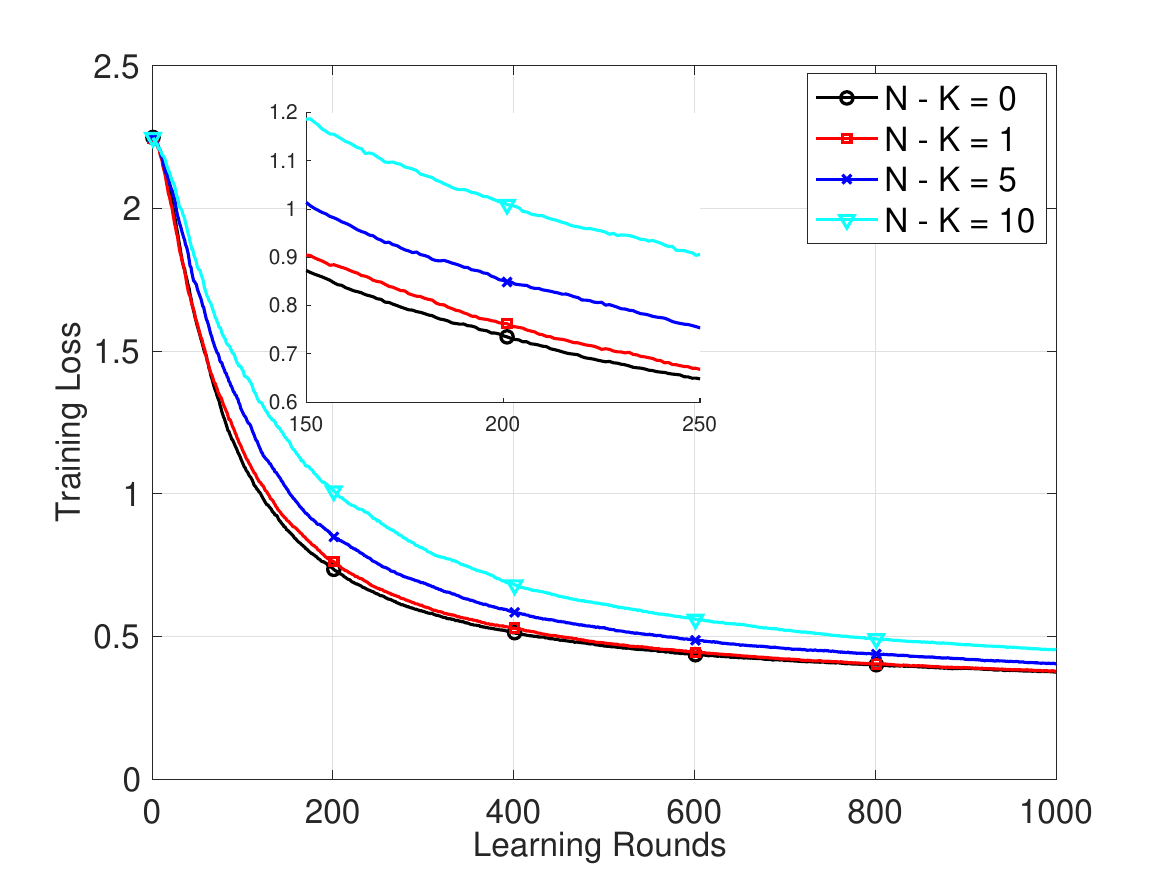}}
    \subfigure[IID MNIST]{\includegraphics[width = 0.45\linewidth]{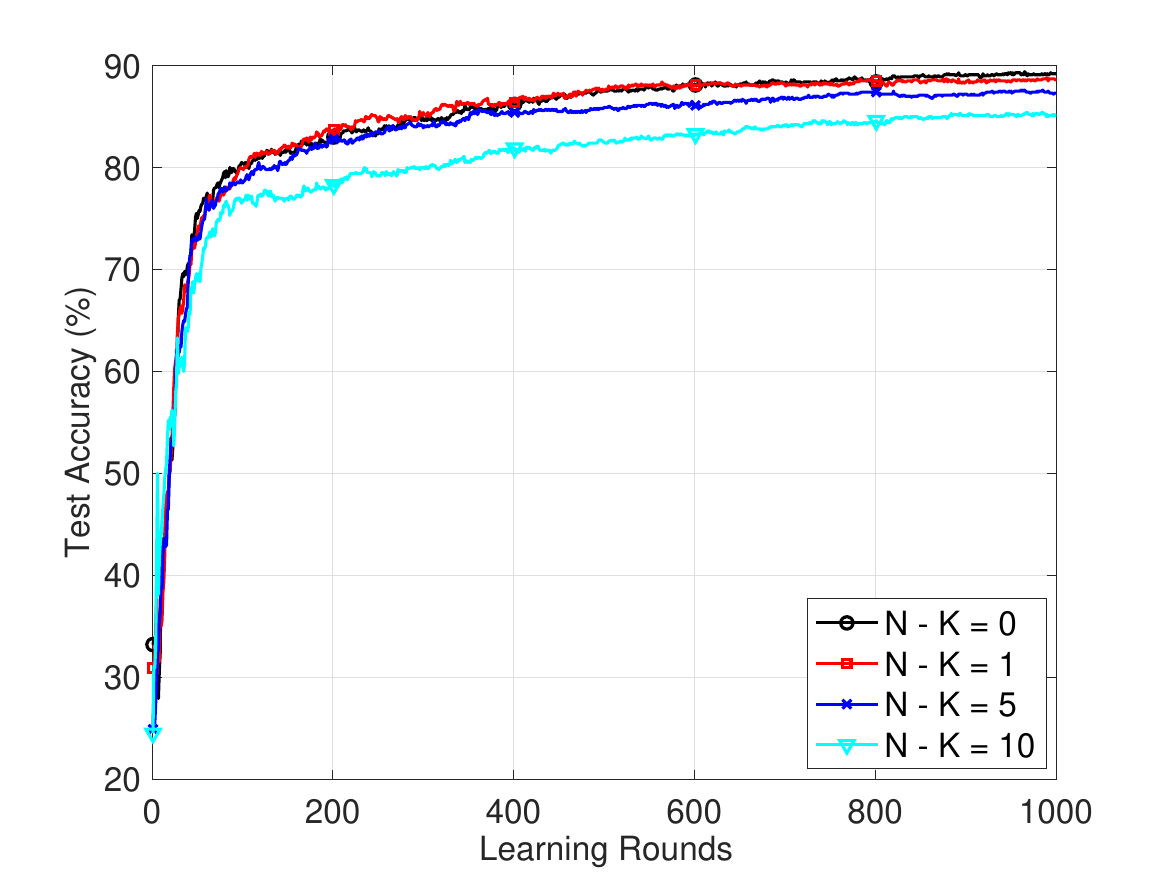}}
    \subfigure[non-IID MNIST]{\includegraphics[width = 0.45\linewidth]{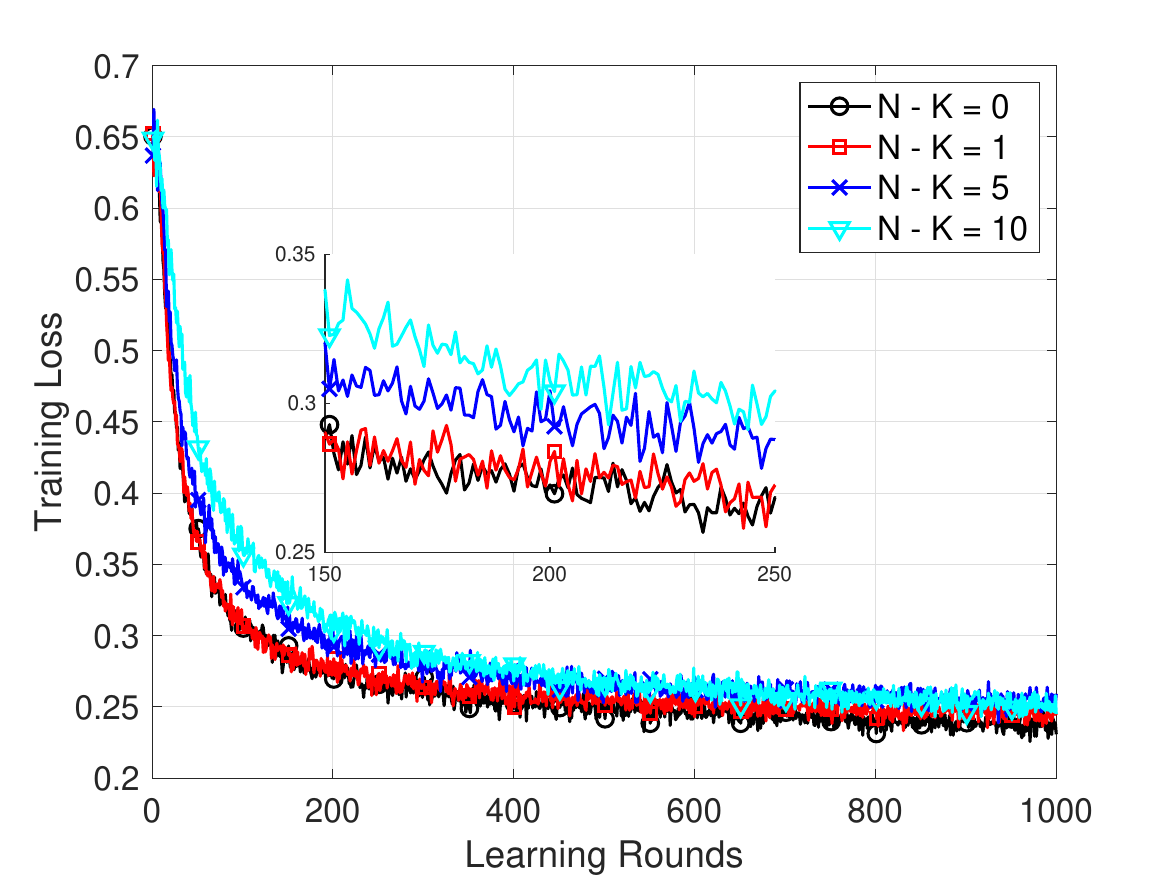}}
    \subfigure[non-IID MNIST]{\includegraphics[width = 0.45\linewidth]{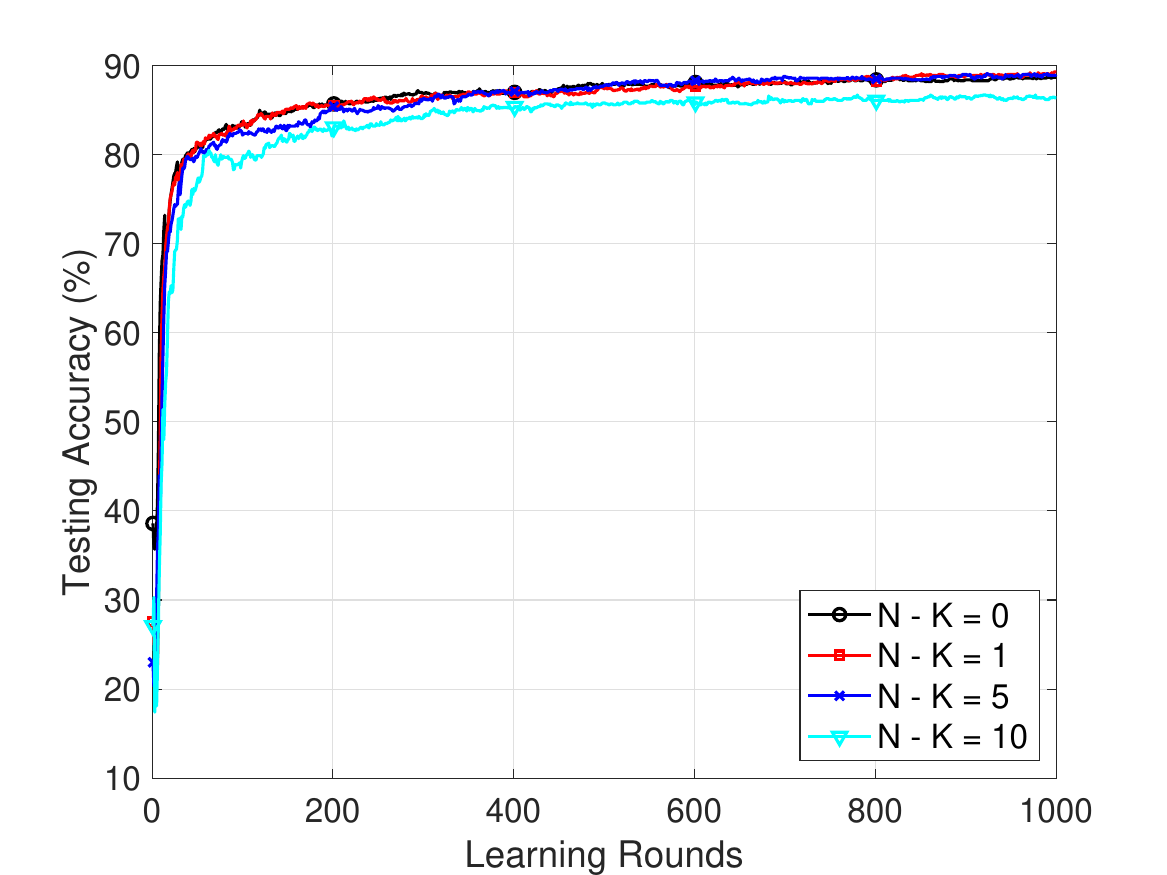}}
    \caption{Training loss/test accuracy versus learning rounds of \alg with different differential privacy levels for MNIST dataset under SNR $= 20$ dB.}
    \label{fig:dp}
\end{figure*}

\subsection{Differential Privacy}
We also evaluate the convergence performance versus different levels of DP. We keep SNR $= 20$ dB and number of selected user $K = 20$, while changing the size of the orthogonal sequence set $\mathcal{A}$ from $20, 21, 25$ to $30$, i.e., $\zixiangbbb{\gamma} = 0, 1, 5$ and $10$. In each learning round, each user selects its own signature from set $\mathcal{A}$ via the method in Section \ref{subsec:AgnmntMchsm}. As discussed in Section \ref{sec:privacy}, the larger size of set $\mathcal{A}$, the higher DP level \alg achieves. The following parameters are used for training: local batch size $20$, the number of local epochs $E=1$, and learning rate $\eta = 0.005$ and $\eta = 0.001$ for IID and non-IID cases, respectively. The training losses with different DP levels are illustrated in Fig.~\ref{fig:dp}-(a) and Fig.~\ref{fig:dp}-(c) of IID and non-IID cases, respectively. It is clear that although a higher privacy level decreases the convergence rate, the machine learning model can still converge with almost the same training loss as the no-differential-privacy case ($\zixiangbbb{\gamma} = 0$), which is consistent with the theoretical analysis in Section \ref{sec:CvgAna}. The test accuracies in Fig.~\ref{fig:dp}-(b) and Fig.~\ref{fig:dp}-(d) further validate the effectiveness of \algg. We can see that the test accuracies for moderate DP levels ($\zixiangbbb{\gamma} = 1$ and $5$) are almost the same as the case of $\zixiangbbb{\gamma} = 0$, i.e., we can achieve certain DP levels \emph{almost for free}. Even when $\zixiangbbb{\gamma} = 10$, the test accuracy loss is tiny, with about $3.5\%$ and $2.5\%$ compared with the $\zixiangbbb{\gamma} = 0$ case in the IID and non-IID datasets, respectively.

\section{Conclusion}\label{sec:concl}

We have proposed \texttt{FLORAS}, a differentially private AirComp FL framework. Compared with the channel inversion method, \alg does not require CSIT and has much more robust performance in the low SNR regime, which is crucial for IoT applications. The flexibility of adjusting the size of the orthogonal sequence set allows us to control both item-level and client-level DP guarantees of the system. From the analyses of convergence and DP of \algg, we have established the tradeoff between convergence rate and privacy preservation, which has been further validated by experiments on real-world FL tasks.

\appendices
% \section{Proof of Rényi DP}\label{app:rdp}
\section{Proofs of Theorems \ref{theo:RDP}, \ref{theo:compositeRDP}, and \ref{theo:UL-RDP}}\label{app:rdp}
\subsection{Lemmas}
We first establish the necessary lemmas for the proofs. % of Theorem \ref{theo:RDP} and \ref{theo:compositeRDP}.
\begin{lemma}\label{lemma:CauchyMinMax}
    Let $P$ and $Q$ be two Cauchy probability distributions $\ssf{Cauchy}(0, \gamma)$ and $\ssf{Cauchy}(\theta, \gamma)$, with PDF $P(x) = \gamma/\pi(x^2 + \gamma^2)$ and $Q(x) = \gamma/\pi[{(x - \theta)^2 + \gamma^2}]$, respectively. $Q(x)/P(x) = (x^2 + \gamma^2)/[(x - \theta)^2 + \gamma^2]$ reaches its maximum at \zixiangbbb{$x_{\max} = \frac{1}{2}(\theta + \sqrt{\theta^2 + 4\gamma^2})$ with $\squab{{Q(x)}/{P(x)}}_{\max} = [\sqrt{\theta^2 + 4\gamma^2} + \theta]/[\sqrt{\theta^2 + 4\gamma^2} - \theta]$, and its minimum at $x_{\min} = \frac{1}{2}(\theta - \sqrt{\theta^2 + 4\gamma^2})$ with $\squab{{Q(x)}/{P(x)}}_{\min} = [\sqrt{\theta^2 + 4\gamma^2} - \theta]/[\sqrt{\theta^2 + 4\gamma^2} + \theta]$.}
\end{lemma}
\begin{proof}
    The results can be directly obtained by taking the first-order derivative of $\frac{Q(x)}{P(x)}$ and setting $\frac{\partial}{\partial x}\frac{Q(x)}{P(x)} = 0$. 
\end{proof}

\begin{lemma}[Proposition $3.3$ in \cite{bun2016concentrated}]\label{lemma:RenyiDPbound}
Let $P$ and $Q$ be probability distributions satisfying $D_\infty (P||Q) \leq \epsilon$ and $D_\infty (Q||P) \leq \epsilon$. Then $D_\alpha (P||Q) \leq \f{1}{2}\alpha\epsilon^2$ for all $\alpha > 1$.
\end{lemma}

\begin{lemma}[Proposition $3$ in \cite{mironov2017renyi}]\label{lemma:CompRuleRDP} If $f$ is an $(\alpha, \epsilon)$-Rényi DP mechanism, it also satisfies $(\epsilon + \f{\log(1/\delta)}{\alpha - 1}, \delta)$-DP for any $0 < \delta < 1$.
\end{lemma}

\subsection{Proof of Theorem \ref{theo:RDP}}\label{app:proofRDP}

Define neighboring datasets \zixiangr{$\mathcal{D} = \bigcup_{i = 1}^M\mathcal{D}_i$ and $\mathcal{D^\prime} = \{\bigcup_{i = 1, i\neq k}^M\mathcal{D}_i\}\cup \mathcal{D}_k^\prime $, where $\mathcal{D}_k^\prime \triangleq \mathcal{D}_k\cup\{z\}$ has one more element than $\mathcal{D}_k$.} As we have assumed that all datasets \zixiangr{$\{\mathcal{D}_i, i = 1,\cdots, M\}$} have the same size $D$,  the size of $\mathcal{D}_k^\prime$ is $D^\prime \triangleq D + 1$. 
% Let $d_{\text{batch}} \triangleq |\xi|$ denote the size of the mini-batch of SGD for all clients. Hence we 
We have $c = \left(\begin{matrix}{D}\\{d_{\text{batch}}}\end{matrix}\right)$ and $c^\prime = \left(\begin{matrix}{D^\prime}\\{d_{\text{batch}}}\end{matrix}\right)$ different mini-batches for dataset \zixiangr{$\{\mathcal{D}_i, i = 1,\cdots, M\}$} and $\mathcal{D}_k^\prime$, respectively. We further denote set $\Xi_{i} \triangleq \{\xi_{i,1}, \cdots,\xi_{i,c}\}$ as a collection of all the mini-batches corresponding to \zixiangr{$\{\mathcal{D}_i, i = 1,\cdots, M\}$}. \zixiangr{We next calculate the number of mini-batches in dataset $\mathcal{D}_k^\prime$.} It is straightforward to verify that
$
    c^\prime = \left(\begin{matrix}{D + 1}\\{d_{\text{batch}}}\end{matrix}\right) = \left(\begin{matrix}{D}\\{d_{\text{batch}}}\end{matrix}\right) + \left(\begin{matrix}{D}\\{d_{\text{batch}} - 1}\end{matrix}\right)$, which reveals that there are $\left(\begin{matrix}{D}\\{d_{\text{batch}}}\end{matrix}\right)$ mini-batches in $\mathcal{D}_k^\prime$ that are the same as those in $\Xi_k$, and the remaining $\left(\begin{matrix}{D}\\{d_{\text{batch}} - 1}\end{matrix}\right)$ mini-batches are the ones that contain the data sample $z$. We denote $\Xi_k^\prime = \{\xi^\prime_{k,1}, \cdots, \xi_{k, c^{\prime} - c}^\prime \}$ as a collection of all the mini-batches that contain the data sample $z$. Therefore, all the possible mini-batches of $\mathcal{D}_k^\prime$ are in $\Xi_k\cup\Xi_k^\prime$.

We note that, \zixiangr{in each learning round, $K$ of total $M$ clients are randomly selected. Therefore, there are total $\chi\triangleq\left(\begin{matrix}M\\K\end{matrix}\right)c^K$ combinations of mini-batches in $\mathcal{D}$. We denote $\mathcal{S}_t =\{s_{t,1},\cdots,s_{t,i},\cdots,s_{t,K}\}$ as the index set for the $K$ selected clients at $t$-th learning round, and it is easy to verify that we have $\left(\begin{matrix}M\\K\end{matrix}\right)$ combinations for client selection per learning round. Therefore, we further denote each possible mini-batch combination as
$\xibf_j \triangleq\{\xi_{s_{t,1},c_1},\cdots,\xi_{s_{t,i},c_i},\cdots, \xi_{s_{t,K},c_K}\}, j = 1,\cdots, \chi,$ where $\xi_{s_{t,i},c_i}\in \Xi_{s_{t,i}}, \forall s_{t,i}\in\mathcal{S}_t$. As for $\mathcal{D}^\prime$, there are total $\left(\begin{matrix}M\\K\end{matrix}\right)c^K + \left(\begin{matrix}M-1\\K-1\end{matrix}\right)(c^\prime - c)c^{K-1}$ mini-batch combinations. Besides the same $\chi$ combinations in $\mathcal{D}$, there are additional $\chi^\prime\triangleq\left(\begin{matrix}M-1\\K-1\end{matrix}\right)(c^\prime - c)c^{K-1}$ mini-batch combinations in $\mathcal{D}^\prime$, which contains the mini-batch of the $k$-th client and is chosen from $\Xi_k^\prime$. We denote them similarly as 
$\xibf^\prime_{j} \triangleq\{\xi_{s_{t,1},c_1},\cdots,\xi_{s_{t,i},c_i},\cdots, \xi_{s_{t,K-1},c_{K-1}},\xi^\prime_{k,c^\prime_k}\}, j = 1,\cdots \chi^\prime$, where $\xi_{s_{t,i},c_i}\in \Xi_{s_{t,i}}, \forall s_{t,i}\in \mathcal{S}_t\backslash[k]$, and $\xi^\prime_{k,c_k^\prime}\in \Xi^\prime_k$.}
% \begin{equation*}
% \begin{split}
%     &\xibf^\prime_{j} \triangleq\{\xi_{1,c_1},\cdots,\xi^\prime_{i,c^\prime_i},\cdots, \xi_{K,c_K}\}\\
%     &\;\;\forall\xi_{i\neq k,c_i}\in \Xi_i,\forall\xi^\prime_{k,c_i^\prime}\in \Xi^\prime_k, j = 1,\cdots (c^\prime -c)c^{K-1}.  
% \end{split}
% \end{equation*}
% Hence there are total $c^\prime c^{K-1}$ mini-batch combinations within $\mathcal{D}^\prime$. 
In addition, we use $g(\xibf_j)$ to denote the noise-free global models calculated from mini-batch combination $\xibf_j$. 

\zixiangr{Building on these, the mechanism $\mathcal{M}(\mathcal{D})$ can be viewed as a random vector sampling from the following distribution: 
% Based on $\Xi$, the mechanism $\mathcal{M}(\mathcal{D})$ samples from a Cauchy distribution with central $g(\xi_i)$. Here, we use $g(\xi_i)$ to denote the distribution of noise-free global models where client $k$ adopts the mini-batch $\xi_i$ from $\mathcal{D}_k$ for local SGD. 
%It can be regarded as a marginal distribution discarding (marginalizing out) the randomness of local SGD from other $k - 1$ clients. Therefore,
%\begin{equation*}
     $$\mathcal{M}(\mathcal{D}) =  \sum_{j = 1}^{\chi} 
    \frac{1}{\chi}\ssf{Cauchy}(g(\xibf_j), \zixiangbbb{\gamma} \Ibf_d).$$
%\end{equation*}
Similarly, 
\begin{align}
\label{eqn:appx2}
    &\mathcal{M}(\mathcal{D}^\prime) = \sum_{j = 1}^{\chi}\frac{1}{\chi+\chi^\prime}\ssf{Cauchy}(g(\xibf_j), \zixiangbbb{\gamma} \Ibf_d) + \sum_{j=1}^{\chi^\prime} \frac{1}{\chi+\chi^\prime}\ssf{Cauchy}(g(\xibf^\prime_j), \zixiangbbb{\gamma} \Ibf_d).
\end{align}
We notice that $\frac{\chi^\prime}{\chi} = \frac{K}{M}\frac{c^\prime - c}{c}$ and $\frac{c^\prime}{c} = 1 + \frac{d_{\text{batch}}}{D+1-d_{\text{batch}}}$. Defining $q\triangleq\frac{d_{\text{batch}}}{D+1-d_{\text{batch}}}$ and $p\triangleq \frac{K}{M}$, $\mathcal{M}(\mathcal{D}^\prime)$ can be re-written as
\begin{align*}
     \mathcal{M}(\mathcal{D}^\prime) 
     = &\frac{1}{\chi+\chi^\prime }\sum_{j = 1}^{\chi} \Big[\ssf{Cauchy}(g(\xibf_j), \zixiangbbb{\gamma} \Ibf_d) + \f{\chi^\prime}{\chi}\ssf{Cauchy}(g(\xibf^\prime_{\lceil j*\f{\chi^\prime}{\chi}\rceil}), \zixiangbbb{\gamma} \Ibf_d)\Big]\\
     = & \frac{1}{\chi}\sum_{j = 1}^{\chi} \Big[\f{\chi}{\chi + \chi^\prime}\ssf{Cauchy}(g(\xibf_j), \zixiangbbb{\gamma} \Ibf_d) + \f{\chi}{\chi + \chi^\prime}\f{\chi^\prime}{\chi}\ssf{Cauchy}(g(\xibf^\prime_{\lceil j*\f{\chi^\prime}{\chi}\rceil}, \zixiangbbb{\gamma} \Ibf_d)\Big]\\
     = & \sum_{j = 1}^{\chi} \frac{1}{\chi}\left[\left(1 - \frac{pq}{1+pq}\right)\ssf{Cauchy}(g(\xibf_j), \zixiangbbb{\gamma} \Ibf_d)+ \frac{pq}{1+pq}\ssf{Cauchy}(\xibf^\prime_{\lceil j*pq\rceil}, \zixiangbbb{\gamma} \Ibf_d)\right],
\end{align*}
where the first equality comes from replicating each of the $\chi^\prime$ items in the second summation in Eqn.~\eqref{eqn:appx2} $\frac{\chi^\prime}{\chi}$ times, thus allowing the summation to be over the same range as the first summation.}
% where the first equality comes from breaking the original summation of $\chi^\prime$ items into $\chi$ groups (each of size $\frac{\chi^\prime}{\chi}$).}

We next bound the Rényi divergence for $\alpha = + \infty$. \zixiangr{Since Rényi divergence is quasi-convex, we have
\begin{align*}\small
    & D_{\infty}(\mathcal{M}(\mathcal{D})||\mathcal{M}(\mathcal{D}^\prime))\\
    \leq & \sup_{j}D_{\infty}\left[\ssf{Cauchy}(g(\xibf_j), \gamma \Ibf_d)||(1 - \frac{pq}{1+pq})\ssf{Cauchy}(g(\xibf_j), \gamma \Ibf_d)+ \frac{pq}{1+pq}\ssf{Cauchy}(g(\xibf^\prime_{\lceil j*pq\rceil}), \gamma\Ibf_d)\right]\\
     \leq &\sup_{j}D_{\infty}\left[\ssf{Cauchy}(0, \gamma \Ibf_d)||(1 - \frac{pq}{1+pq})\ssf{Cauchy}(0, \gamma \Ibf_d)+ \frac{pq}{1+pq}\ssf{Cauchy}(g(\xibf^\prime_{\lceil j*pq\rceil}) - g(\xibf_j), \gamma\Ibf_d)\right]. 
\end{align*}
As shown in \eqref{eq:GSbound}, $\norm{g(\xibf^\prime_{\lceil j*pq\rceil}) - g(\xibf_j)}_2\leq 2C$. Therefore, via a rotation, we have that $g(\xibf^\prime_{\lceil j*pq\rceil}) - g(\xibf_j) = c_\xi \ebf_1$ where $c_\xi \leq 2C$. By the additivity of Rényi divergence for product distributions \cite{van2014renyi}, we have that
\begin{align*}\small
     D_{\infty}(\mathcal{M}(\mathcal{D})||\mathcal{M}(\mathcal{D}^\prime))
     \leq  D_{\infty}\left[\ssf{Cauchy}(0, \gamma )||(1 - \frac{pq}{1+pq})\ssf{Cauchy}(0, \gamma ) + \frac{pq}{1+pq}\ssf{Cauchy}(2C, \gamma)\right].
\end{align*}

We next bound $D_{\infty}(\mathcal{M}(\mathcal{D})||\mathcal{M}(\mathcal{D}^\prime))$ and $D_{\infty}(\mathcal{M}(\mathcal{D}^\prime)||\mathcal{M}(\mathcal{D}))$. Based on the results in Lemma \ref{lemma:CauchyMinMax}, we have  
\begin{equation*}
    \max\squab{\frac{\ssf{Cauchy}(2C, \gamma)}{\ssf{Cauchy}(0, \gamma)} } = \zixiangbbb{\f{\sqrt{C^2 + \gamma^2} + C}{\sqrt{C^2 + \gamma^2} - C}},
\end{equation*}
and
\begin{equation*}
    \min\squab{\frac{\ssf{Cauchy}(2C, \gamma)}{\ssf{Cauchy}(0, \gamma)} } =  \zixiangbbb{\f{\sqrt{C^2 + \gamma^2} - C}{\sqrt{C^2 + \gamma^2} + C}}.
\end{equation*}
Therefore, $D_{\infty}(\mathcal{M}(\mathcal{D})||\mathcal{M}(\mathcal{D}^\prime))$ and $D_{\infty}(\mathcal{M}(\mathcal{D}^\prime)||\mathcal{M}(\mathcal{D}))$ can be bounded respectively as follows:
\begin{align*}\small
    &  D_{\infty}\left[(1 - \frac{pq}{1+pq})\ssf{Cauchy}(0, \gamma ) + \frac{pq}{1+pq}\ssf{Cauchy}(2C, \gamma)||\ssf{Cauchy}(0, \gamma )\right]\\
     = & \sup\log \left( 1 - \frac{pq}{1+pq} + \frac{pq}{1+pq}  \frac{\ssf{Cauchy}(2C, \gamma)}{\ssf{Cauchy}(0, \gamma)} \right)   \\
      \leq  &  \log\left(1 + \zixiangbbb{\frac{pq}{1+pq}\frac{ 2C\sqrt{C^2 + \gamma^2} + 2C^2 }{\gamma^2}}\right)
\end{align*}
and
\begin{align*}\small
    &  D_{\infty}\left[\ssf{Cauchy}(0, \gamma )||(1 - \frac{pq}{1+pq})\ssf{Cauchy}(0, \gamma ) + \frac{pq}{1+pq}\ssf{Cauchy}(2C, \gamma)\right]\\
     = &\sup\log \left( \f{1}{1 - \frac{pq}{1+pq} + \frac{pq}{1+pq}  \frac{\ssf{Cauchy}(2C, \gamma)}{\ssf{Cauchy}(0, \gamma)} }\right) 
     \leq  \log\left(1 + \zixiangbbb{\frac{pq}{1+pq}\frac{ 2C\sqrt{C^2 + \gamma^2} - 2C^2 }{\gamma^2}}\right).
\end{align*}
\zixiangbbb{It is straightforward to verify that $D_{\infty}(\mathcal{M}(\mathcal{D})||\mathcal{M}(\mathcal{D}^\prime))\geq D_{\infty}(\mathcal{M}(\mathcal{D}^\prime)||\mathcal{M}(\mathcal{D})), \forall p,q, \text{ such that } 0<\frac{pq}{1+pq}<1$. Therefore, both $D_{\infty}(\mathcal{M}(\mathcal{D})||\mathcal{M}(\mathcal{D}^\prime))$ and $D_{\infty}(\mathcal{M}(\mathcal{D}^\prime)||\mathcal{M}(\mathcal{D}))$  can be bounded by $\log\big(1 + \zixiangbbb{\frac{pq}{1+pq}\frac{ 2C\sqrt{C^2 + \gamma^2} + 2C^2 }{\gamma^2}}\big)$. 
Based on Lemma \ref{lemma:RenyiDPbound}, we can guarantee $D_{\alpha}(\mathcal{M}(\mathcal{D})||\mathcal{M}(\mathcal{D}^\prime))\leq \f{1}{2}\alpha\log^2\big(1 + \zixiangbbb{\frac{pq}{1+pq}\frac{ 2C\sqrt{C^2 + \gamma^2} + 2C^2 }{\gamma^2}}\big)$, which completes the proof.}}

\subsection{Proof of Theorem \ref{theo:compositeRDP}}\label{app:ProofCompositeRDP}
Based on Theorem \ref{theo:RDP} and the composition rule of Rényi DP as shown in Lemma \ref{lemma:CompRuleRDP}, the Rényi DP guarantee for total $T$ learning rounds is \zixiangr{$\squab{\f{1}{2}T\alpha\log^2\left(1 + \zixiangbbb{\frac{pq}{1+pq}\frac{ 2C\sqrt{C^2 + \gamma^2} + 2C^2 }{\gamma^2}}\right), \alpha}$-Rényi DP}. Therefore, the proposed design satisfies $(\epsilon^\prime,\delta)$-DP, where
% \begin{equation*}
% \resizebox{\linewidth}{\height}{%
% $
\zixiangr{
\begin{align*}
 \epsilon^\prime & = \min_{\alpha > 1} \f{1}{2}T\alpha\log^2\left(1 + \zixiangbbb{\frac{pq}{1+pq}\frac{ 2C\sqrt{C^2 + \gamma^2} + 2C^2 }{\gamma^2}}\right) + \f{\log(1/\delta)}{\alpha - 1} \\
 & \geq \sqrt{2T\log(\frac{1}{\delta})}\log\left(1 + \zixiangbbb{\frac{pq}{1+pq}\frac{ 2C\sqrt{C^2 + \gamma^2} + 2C^2 }{\gamma^2}}\right) + \f{1}{2}T\log^2\left(1 + \zixiangbbb{\frac{pq}{1+pq}\frac{ 2C\sqrt{C^2 + \gamma^2} + 2C^2 }{\gamma^2}}\right).
\end{align*}}
% $
% } 
% \end{equation*}

\subsection{Proof of Theorem \ref{theo:UL-RDP}}\label{app:proofULRDP}
Define neighboring datasets \zixiangr{$\mathcal{D} = \bigcup_{i = 1}^M\mathcal{D}_i$} and \zixiangr{$\mathcal{D^\prime} = \bigcup_{i = 1}^M\mathcal{D}_i\cup \mathcal{D}_k^\prime $}, \zixiangr{where $\mathcal{D}_k^\prime$ is an arbitrary local dataset of a client $k$}. \zixiangr{For dataset $\mathcal{D}$, $K$ of the total $M$ clients are randomly scheduled for the task in each learning round. We denote $\mathcal{S}_j, \forall j=1,\cdots,\psi$, where $\psi \triangleq \Big(\begin{matrix}M\\K\end{matrix}\Big)$, as the set of all the possible client combinations during a single learning round for dataset $\mathcal{D}$. For dataset $\mathcal{D}^\prime$, beside the previous $\psi$ combinations, there are additional $\psi^\prime \triangleq \Big(\begin{matrix}M\\K-1\end{matrix}\Big)$ combinations that contain client $k$. We denote them as $\mathcal{S}_j^\prime, \forall j = 1,\cdots,\psi^\prime$.} Therefore, the mechanism $\mathcal{M}(\mathcal{D})$ samples from the following distribution
\zixiangr{
\begin{equation*}
    \mathcal{M}(\mathcal{D}) = \frac{1}{\psi}\sum_{j = 1}^\psi \ssf{Cauchy}(g(\mathcal{S}_j), \zixiangbbb{\gamma} \Ibf_d).
\end{equation*}
Similarly, 
\begin{align*}
    \mathcal{M}(\mathcal{D}^\prime)  & = \frac{1}{\psi + \psi^\prime}\sum_{j = 1}^{\psi} \ssf{Cauchy}(g(\mathcal{S}_j), \zixiangbbb{\gamma} \Ibf_d) + \frac{1}{\psi+\psi^\prime} \sum_{j = 1}^{\psi^\prime} \ssf{Cauchy}(g(\mathcal{S}^\prime_j), \zixiangbbb{\gamma} \Ibf_d)\\
    & = \frac{1}{\psi}\Big[\frac{\psi}{\psi + \psi^\prime}\sum_{j = 1}^{\psi} \ssf{Cauchy}(g(\mathcal{S}_j), \zixiangbbb{\gamma} \Ibf_d) + \frac{\psi}{\psi+\psi^\prime} \frac{\psi^\prime}{\psi} \ssf{Cauchy}(g(\mathcal{S}^\prime_j), \zixiangbbb{\gamma} \Ibf_d)\Big]\\
    & = \frac{1}{\psi}\Big[\left(1-\frac{K}{M + 1}\right)\sum_{j = 1}^{\psi} \ssf{Cauchy}(g(\mathcal{S}_j), \zixiangbbb{\gamma} \Ibf_d) + \frac{K}{M + 1} \ssf{Cauchy}(g(\mathcal{S}^\prime_j), \zixiangbbb{\gamma} \Ibf_d)\Big].
\end{align*}
By the definition of $p\triangleq \frac{K}{M+1}$ and following the similar techniques in Appendix \ref{app:proofRDP}, $D_{\infty}(\mathcal{M}(\mathcal{D})||\mathcal{M}(\mathcal{D}^\prime))$ and $D_{\infty}(\mathcal{M}(\mathcal{D}^\prime)||\mathcal{M}(\mathcal{D}))$ can both be bounded by
$\log\left(1 + \zixiangbbb{p\frac{ 2C\sqrt{C^2 + \gamma^2} + 2C^2 }{\gamma^2}}\right)$. Again, based on Lemma \ref{lemma:RenyiDPbound}, we can guarantee $D_{\alpha}(\mathcal{M}(\mathcal{D})||\mathcal{M}(\mathcal{D}^\prime))\leq \f{1}{2}\alpha\log^2\left(\zixiangbbb{1 + p\frac{ 2C\sqrt{C^2 + \gamma^2} + 2C^2 }{\gamma^2}}\right)$, which completes the single round DP guarantee. The proof of the client-level DP guarantee for the composition of $T$ rounds follows the same as that in Appendix \ref{app:ProofCompositeRDP}.}

\section{Proof of Theorem \ref{thm.Conv}}\label{App:thm.Conv}
With a slight abuse of notation, we change the timeline to be with respect to the overall SGD iteration time steps instead of the communication rounds, i.e., $$t=\underbrace{1, \cdots, E}_{\text{round 1}}, \underbrace{E+1, \cdots, 2E}_{\text{round 2}}, \cdots, \cdots, \underbrace{(T-1)E+1, \cdots, TE}_{\text{round $T$}}.$$ 
Note that the global model $\vect{w}_{t}$ is only accessible at the clients for specific $t \in \mathcal{I}_E$, where $\mathcal{I}_E=\{nE ~|~n=1,2,\dots\}$, i.e., the time steps for communication.  The notations for $\eta_t$ are similarly adjusted to this extended timeline, but their values remain constant within the same round. The key technique in the proof is the {\em perturbed iterate framework}  in \cite{mania2017siam}. In particular, we first define the following variables for client $k\in [M]$: 
\begin{align*}
    \vect{v}_{t+1}^k & \triangleq \vect{w}_t^k - \eta_t \nabla \tilde f_k(\vect{w}_t^k); \\
    \vect{u}_{t+1}^k & \triangleq \begin{cases}
        \vect{v}_{t+1}^k & \text{if~} t+1 \notin \mathcal{I}_E, \\
        \frac{1}{K} \sum_{k = 1}^K \vect{v}_{t+1}^i   & \text{if~} t+1 \in \mathcal{I}_E; 
    \end{cases} \\
    \vect{w}_{t+1}^k & \triangleq \begin{cases}
         \vect{v}_{t+1}^k & \text{if~} t+1 \notin \mathcal{I}_E, \\
          \vect{u}_{t+1}^k + \frac{1}{K}\nbf_{t + 1}
         & \text{if~} t+1 \in \mathcal{I}_E;
    \end{cases}
 \end{align*}
 \zixiangbbb{where $\nbf_{t + 1} \triangleq\frac{\zixiangr{C_{\max,t}}}{C}\squab{\sum_{k = K + 1}^N \frac{\abf_k^T\nbf_1}{\abf^T_k \nbf_s}, \cdots, \sum_{k = K + 1}^N \frac{\abf_k^T\nbf_d}{\abf^T_k \nbf_s}}^T\in \mathbb{C}^{d\times 1}$ is the effective noise vector after de-normalization. Note that $\nbf_{t+1}$ is a truncated Cauchy distribution vector with the following PDF:
 \begin{equation*}
 f(n_{t+1,i}) =  \frac{\gamma}{\left (n_{t+1,i}^2 + {\gamma}^2 \right )\left( \arctan\squab{\frac{B -\sum_{k = 1}^K x_k^i}{\zixiangbbb{\gamma}}} + \arctan\squab{\frac{B + \sum_{k = 1}^K x_k^i}{\gamma}}\right)},
 \end{equation*}
 % \begin{equation*}\small
 % \begin{split}
 %      f(n_{t+1,i}) =  & \frac{1}{\arctan\squab{\frac{B -\sum_{k = 1}^K x_k^i}{\zixiangbbb{\gamma}}} + \arctan\squab{\frac{B + \sum_{k = 1}^K x_k^i}{\zixiangbbb{\gamma}}}} \times \\
 %     & \frac{\zixiangbbb{\gamma}}{n_{t+1,i}^2 + \zixiangbbb{\gamma}^2},
 % \end{split}
 % \end{equation*}
where $n_{t+1,i}\in\left[-B -\sum_{k = 1}^K x_k^i,B -\sum_{k = 1}^K x_k^i\right]$, $\forall i = 1,\cdots,d$.
We further have 
 \begin{equation*}\small
\begin{split}
    \expt \norm{\frac{\zixiangr{C_{\max,t}}}{C}n_{t+1,i}}^2 
    & = \frac{\zixiangr{C^2_{\max,t}}}{C^2}\squab{\frac{2\zixiangbbb{\gamma}B}{\arctan\squab{\frac{B -\sum_{k = 1}^K x_k^i}{\zixiangbbb{\gamma}}} + \arctan\squab{\frac{B + \sum_{k = 1}^K x_k^i}{\zixiangbbb{\gamma}}}} - \zixiangbbb{\gamma}^2}\\ 
    & \leq \frac{\zixiangbbb{\gamma}^2\zixiangr{C_{\max,t}}^2}{C^2\arctan\left(\frac{B+C}{\zixiangbbb{\gamma}}\right)}\squab{\frac{B}{\zixiangbbb{\gamma}} - \arctan\left(\frac{B+C}{\zixiangbbb{\gamma}}\right)}\\
    & \triangleq \zixiangr{\zixiangr{C^2_{\max,t}}D(\gamma)}.
\end{split}
\end{equation*}
 }
Then, we construct the following \textit{virtual sequences}:
$$\avgvect{v}_t=\frac{1}{M}\sum_{k=1}^M \vect{v}_t^k,\;\;\avgvect{u}_t=\frac{1}{M}\sum_{k=1}^M \vect{u}_t^k,\;\;\text{and~~}\avgvect{w}_t=\frac{1}{M}\sum_{k=1}^M \vect{w}_t^k.$$ 
We also define $\avgvect{g}_{t} = \frac{1}{M}\sum_{k=1}^M \nabla f_k(\vect{w}_t^k)$ and $\vect{g}_{t} = \frac{1}{M}\sum_{k=1}^M \nabla \tilde f_k(\vect{w}_t^k)$ for convenience. Therefore, $\avgvect{v}_{t+1} = \avgvect{w}_t - \eta_t \vect{g}_t$ and $\expt \squab{\vect{g}_t} = \avgvect{g}_t$. Note that the global model $\wbf_{t + 1}$ is only meaningful when $t+1 \in \mathcal{I}_E$. Hence, we have $\vect{w}_{t + 1}  =\frac{1}{K}\sum_{k=1}^K \vect{w}_{t + 1}^k = \vect{w}_{t + 1}^k = \frac{1}{M}\sum_{k=1}^K \vect{w}_{t + 1}^k = \avgvect{w}_{t + 1}$.
Thus it is sufficient to analyze the convergence of $\norm{\avgvect{w}_{t+1}-\vect{w}^*}^2$ to evaluate \algg.
\subsection{Lemmas}
% We first establish the following necessary lemmas that are useful in the proof of Theorem \ref{thm.Conv}. 
\label{sec:lemmas}
\begin{lemma}%[Bounding one-step SGD]
\label{lemma:one-step-sgd}
Let Assumptions 1-4 hold, $\eta_t$ is non-increasing, and $\eta_t \leq 2 \eta_{t+E}$ for all $t \geq 0$. If $\eta_t \leq {1}/(4L)$, we have
\begin{equation*}
%\begin{split}
    \expt \norm{\avgvect{v}_{t+1} - \vect{w}^*}^2  \leq (1-\eta_t \mu) \expt \norm{\avgvect{w}_t - \vect{w}^*}^2 
     + \eta_t^2 \left({\sum_{k = 1}^M H_k^2}/{M^2} + 6 L \Gamma + 8(E-1)^2H^2\right). 
%\end{split}
\end{equation*}
\end{lemma}
Lemma~\ref{lemma:one-step-sgd} establishes a bound for the one-step SGD. This result only concerns the local model update and is not impacted by the noisy communication. The derivation is similar to the technique in \cite{stich2018local} and is omitted. % and is omitted due to space limitations.

\begin{lemma}%[Unbiased and variance bounded client sampling]
\label{lemma:sample}
    Let Assumptions 1-4 hold. With $\eta_t \leq 2\eta_{t+E}$ for all $t \geq 0$ and $\forall t+1 \in \mathcal{I}_E$, we have
     \begin{equation*}\label{eqn:sam_unbiased}
        \expt \squab{\avgvect{u}_{t+1}} = \avgvect{v}_{t+1}, \quad  \expt \norm{\avgvect{v}_{t+1} - \avgvect{u}_{t+1}}^2 \leq \frac{M-K}{M-1} \frac{4}{K} \eta_t^2 E^2 H^2.
     \end{equation*}
     % and
     % \begin{equation*}
     %    \expt \norm{\avgvect{v}_{t+1} - \avgvect{u}_{t+1}}^2 \leq \frac{M-K}{M-1} \frac{4}{K} \eta_t^2 E^2 H^2.
     % \end{equation*}
\end{lemma}
\begin{proof}
    Let $\mathcal{S}_{t+1}$ denote the set of chosen indexes. Note that the number of possible $\mathcal{S}_{t+1}$ is $C_N^K$ and we denote the $l$th possible result as $\mathcal{S}_{t+1}^l = \{i_1^l, \dots, i_K^l\}$, where $l=1,\dots,C_N^K$. Therefore,
    \begin{equation*}
       \sum_{j=1}^{C_N^K} \sum_{k=1}^K \vect{v}_{t+1}^{i_k^l} = \frac{K \cdot C_N^K}{N} \sum_{i=1}^N {\vect{v}_{t+1}^k} 
    = C_{N-1}^{K-1}\sum_{i=1}^N {\vect{v}_{t+1}^k}.
    \end{equation*}
    Since when $t+1 \in \mathcal{I}_E$, $$\vect{u}_{t+1}^k = \frac{1}{K} \sum_{k \in S{t+1}} \vect{v}_{t+1}^k$$ for all $k$, we have
    \begin{equation*}
   \avgvect{u}_{t+1} = \sum_{k=1}^N \vect{u}_{t+1}^k = \frac{1}{K} \sum_{k \in S_{t+1}} \vect{v}_{t+1}^k.
    \end{equation*}
    Then
    \begin{equation*}
    \begin{split}
        \expt_{\mathcal{S}_{t}}\squab{\avgvect{u}_{t+1}} & = \sum_{l=1}^{C_N^K} \mathbb{P}\left (\mathcal{S}_{t+1} \mathcal{S}_{t+1}^l \right )  \frac{1}{K} \sum_{k \in S_{t+1}^l} \vect{v}_{t+1}^k \\
        & = \frac{1}{C_N^K} \frac{1}{K} \sum_{j=1}^{C_N^K} \sum_{k=1}^K \vect{v}_{t+1}^{i_k^l} \\
        & = \frac{C_{N-1}^{K-1}}{C_N^K} \frac{1}{K}\sum_{k=1}^N {\vect{v}_{t+1}^k} \\
        & = \frac{1}{N} \sum_{k=1}^N {\vect{v}_{t+1}^k} \\
        & = \avgvect{v}_{t+1}.
    \end{split}
    \end{equation*}
    As for the variance, we have \cite{li2019convergence}
    \begin{equation}
    \label{eqn:uv_st}
        \begin{split}
           & \expt_{\mathcal{S}_{t}}  \norm{\avgvect{u}_{t+1}-\avgvect{v}_{t+1}}^2  = \expt_{\mathcal{S}_{t}} \norm{\frac{1}{K} \sum_{i \in S_{t+1}} \vect{v}_{t+1}^{i}-\avgvect{v}_{t+1}}  \\
           & =  \frac{1}{K^2} \expt_{\mathcal{S}_{t}} \norm{\sum_{i=1}^{N} \mathbb{I}\left\{i \in S_{t}\right\}\left(\vect{v}_{t+1}^{i}-\avgvect{v}_{t+1}\right)}^2 \\
           & =  \frac{1}{K^{2}}\left[\sum_{i \in[N]} \mathbb{P}\left(i \in S_{t+1}\right) \norm{\vect{v}_{t+1}^{i}-\avgvect{v}_{t+1}}^2 \right. \left.+ \sum_{i \neq j} \mathbb{P}\left(i, j \in S_{t+1}\right) \dotp{\vect{v}_{t+1}^{i}-\avgvect{v}_{t+1}}{\vect{v}_{t+1}^{j}-\avgvect{v}_{t+1}} \right ] \\
            &= \frac{1}{K N} \sum_{i=1}^{N} \norm{\vect{v}_{t+1}^{i}-\avgvect{v}_{t+1}}^2  + \sum_{i \neq j} \frac{K-1}{K N(N-1)} \dotp{\vect{v}_{t+1}^{i}-\avgvect{v}_{t+1}}{\vect{v}_{t+1}^{j}-\avgvect{v}_{t+1}} \\
            & =  \frac{1-\frac{K}{N}}{K(N-1)} \sum_{i=1}^{N}\norm{\vect{v}_{t+1}^{i}-\avgvect{v}_{t+1}}^2
        \end{split}
    \end{equation}
    where we use the following results: $$\mathbb{P}\left(i \in S_{t+1}\right) = \frac{K}{N}$$ and $$\mathbb{P}\left(i,j \in S_{t+1}\right) = \frac{K(K-1)}{N(N-1)}$$ for all $i \neq j$, and $$\sum_{i \in [N]} \norm{\vect{v}_{t+1}^i -\avgvect{v}_{t+1}}^2 + \sum_{i \neq j} \dotp{\vect{v}_{t+1}^{i}-\avgvect{v}_{t+1}}{\vect{v}_{t+1}^{j}-\avgvect{v}_{t+1}} = 0.$$    
    Since $t+1 \in \mathcal{I}_E$, we know that $t_0=t-E+1 \in \mathcal{I}_E$ is the communication time, implying that $\{\vect{u}_{t_0}^k \}_{k=1}^N$ are identical. Then
    \begin{equation*}
        \begin{split}
            & \sum_{i=1}^N \norm{\vect{v}_{t+1}^i - \avgvect{v}_{t+1}}^2  = \sum_{i=1}^N \norm{(\vect{v}_{t+1}^i - \avgvect{u}_{t_0}) - (\avgvect{v}_{t+1} - \avgvect{u}_{t_0})}^2 \\
            & = \sum_{i=1}^N \norm{\vect{v}_{t+1}^i - \avgvect{u}_{t_0}}^2 - 2 \dotp{\sum_{i=1}^N \vect{v}_{t+1}^i - \avgvect{u}_{t_0}}{\avgvect{v}_{t+1}-\avgvect{u}_{t_0}} + \sum_{i=1}^N \norm{\avgvect{v}_{t+1} - \avgvect{u}_{t_0}}^2 \\
            & = \sum_{i=1}^N \norm{\vect{v}_{t+1}^i - \avgvect{u}_{t_0}}^2 - \sum_{i=1}^N \norm{\avgvect{v}_{t+1} - \avgvect{u}_{t_0}}^2 \\
            & \leq \sum_{i=1}^N \norm{\vect{v}_{t+1}^i - \avgvect{u}_{t_0}}^2
        \end{split}
    \end{equation*}
    Taking expectation over the randomness of stochastic gradient on Eqn.~\eqref{eqn:uv_st}, we have
    \begin{equation*}
        \begin{split}
            & \expt  \squab{\frac{1}{K(N-1)}\left(1-\frac{K}{N}\right) \sum_{k=1}^N \norm{\vect{v}_{t+1}^i - \avgvect{v}_{t+1}}^2} \\
            & \leq \frac{N-K}{K(N-1)} \frac{1}{N} \sum_{k=1}^N \expt \norm{\vect{v}_{t+1}^i - \avgvect{u}_{t_0}}^2 \\
            & \leq \frac{N-K}{K(N-1)}\frac{1}{N} \sum_{k=1}^N E \sum_{i=t_0}^t \expt \norm{\eta_i \nabla F_k{(\vect{u}_i^k, \xi_i^k)}}^2 \\
            &\leq \frac{N-K}{K(N-1)} E^2 \eta_{t_0}^2 H^2 \\
            & \leq \frac{N-K}{N-1}\frac{4}{K} E^2 \eta_t^2 H^2
        \end{split}
    \end{equation*}
    where in the last line is because $\eta_t$ is non-increasing and $\eta_{t_0} \leq 2\eta_t$.
\end{proof}
Lemma~\ref{lemma:sample} demonstrates a bound for the uniformly random selection. This result bounded the extra ``noise" brought by \emph{partial participation}. We note that $\expt \norm{\avgvect{v}_{t+1} - \avgvect{u}_{t+1}}^2 = 0$ in the full participation case, where $K = M$. We note that the derivation of the proof is similar to the technique in \cite{stich2018local}. %and is omitted. % and is omitted due to space limitations.

\begin{lemma}\label{lemma:ROUL}
Let Assumptions 1-4 hold. With $\eta_t \leq 2\eta_{t+E}$ for all $t \geq 0$ and $\forall t+1 \in \mathcal{I}_E$, we have 
%\begin{equation*}\small
%\begin{split}
    $$\expt \norm{\avgvect{w}_{t+1} - \avgvect{u}_{t + 1}}^2  \leq \frac{4dD(\gamma)}{K^2} \eta^2_{t} E^2 H^2.$$
%\end{split}
%\end{equation*}
\end{lemma}
\begin{proof}
As shown previously, $\nbf_{t + 1}$ is a truncated Cauchy random vector. \zixiangr{We have
    $\expt[\nbf^H_{t + 1}\nbf_{t + 1}] \leq d\zixiangr{C^2_{\max,t}}D(\gamma)$.}
%\end{split}
%\end{equation*}$
%by utilizing the results in \cite{nadarajah2006truncated}.
 %\\
     % & = \frac{1}{K} \sum_{k = 1}^K \vect{v}_{t+1}^i + \expt\squab{ \frac{1}{K}\nbf_{t + 1}}  = \avgvect{u}_{t + 1}.
%\end{split}
%\end{equation*}
%\end{align*}
\zixiangr{
Recall the definition that $\zixiangr{C_{\max,t}} \triangleq \max\{\norm{{\xbf}_t^k - \mu_k}_2,\forall k\}$, which achieves its maximum at $k_{\max}$, and we use $\mu_{k_{\max}}$ to denote the element-wise mean of $\xbf_t^{k_{\max}}$. %Therefore, we have 
%\begin{equation*}
%\begin{split}
%    \expt[\zixiangr{C_{\max,t}}] & = \expt\norm{\xbf_t^{k_{\max}} - \mu_{k_{\max}}}^2 = \variance[\xbf_t^{k_{\max}}] \\&= \expt\norm{\xbf_t^{k_{\max}}}^2 -  \expt[\mu_{k_{\max}}]^2 \leq \expt\norm{\xbf_t^{k_{\max}}}^2.
%\end{split}
%\end{equation*}
\zixiangbbb{Since
%\begin{equation*}
%\begin{align*}
%\begin{split}
     $\avgvect{w}_{t+1}  = \frac{1}{M}\sum_{k=1}^M \vect{w}_{t + 1}^k =  \avgvect{u}_{t + 1} + \frac{1}{K}\nbf_{t+1}$}, we can bound $\expt \norm{\avgvect{w}_{t+1} - \avgvect{u}_{t + 1}}^2$ as

\begin{align*}%\small
%\begin{split}
\label{eq.A1}
     \expt \norm{\avgvect{w}_{t+1} - \avgvect{u}_{t + 1}}^2 &= \expt\left\Vert \frac{1}{K} \nbf_{t + 1} \right\Vert^2 \\
     & \leq \expt \squab{\frac{d\zixiangr{C^2_{\max,t}}D(\gamma)}{K^2} } \\
     & = \frac{dD(\gamma)}{K^2} \expt\norm{\xbf_t^{k_{\max}} - \mu_{k_{\max}}}^2 \\
    & \leq \frac{dD(\gamma)}{K^2} \expt\norm{\xbf_t^{k_{\max}} }^2 \\
    & \leq \frac{dD(\gamma)}{K^2} E \sum_{i = t + 1 - E}^t\expt\norm{\eta_i \nabla\tilde{f}_{k_{\max}}(\wbf_i^{k_{\max}})}^2 \\
    & \leq \frac{dD(\gamma)}{K^2} \eta^2_{t+1-E} E^2 H^2 \\
    & \leq \frac{4dD(\gamma)}{K^2} \eta^2_{t} E^2 H^2,
%\end{split}
\end{align*}}
where in the last inequality we use the fact that $\eta_t$ is non-increasing and $\eta_{t + 1 - E}\leq 2\eta_t$.
\end{proof}

\subsection{Proof of Theorem \ref{thm.Conv}}
\label{sec:proofThrm1}
We next consider the convergence of $\expt\norm{\avgvect{w}_{t+1}-\vect{w}^*}^2$.

1) If $t+1 \notin \mathcal{I}_E$, $\avgvect{v}_{t + 1} = \avgvect{w}_{t + 1}$. Using Lemma \ref{lemma:one-step-sgd}, we have:
\begin{equation*}\small
    \label{eqn:case1}
    \begin{split}
    & \expt  \norm{\avgvect{w}_{t+1} - \vect{w}^*}^2  = \expt \norm{\avgvect{v}_{t+1} - \vect{w}^*}^2 \leq (1-\eta_t \mu) \expt \norm{\avgvect{w}_t - \vect{w}^*}^2   + \eta_t^2 \squab{\sum_{k=1}^{M}\frac{H_k^2}{M^2} + 6L \Gamma + 8(E-1)^2 H^2}.
    \end{split}
\end{equation*}
        % \begin{align}
        %   &  \expt  \norm{\avgvect{p}_{t+1} - \vect{w}^*}^2 = \expt \norm{\avgvect{v}_{t+1} - \vect{w}^*}^2 \nonumber \\
        %     % & \leq (1-\eta_t \mu) \expt \norm{\avgvect{w}_t - \vect{w}^*}^2 + \eta_t^2 \expt \norm{\vect{g}_t - \avgvect{g}_t}^2 + 6 L \eta_t^2 \Gamma + 2\expt \squab{\frac{1}{N}\sum_{k=1}^N \norm{\avgvect{w}_t - \vect{w}_t^k}^2} \nonumber \\
        %     & \leq (1-\eta_t \mu) \expt \norm{\avgvect{p}_t - \vect{w}^*}^2 + \eta_t^2 \squab{\sum_{k=1}^{N}\frac{\delta_k^2}{N^2} + 6L \Gamma + 8(E-1)^2 H^2}.
        %     \label{eqn:case1}
        % \end{align}

2) If $t+1 \in \mathcal{I}_E$, to evaluate the convergence of $\expt\norm{\avgvect{w}_{t+1}-\vect{w}^*}^2$, we establish
 \begin{equation*} \label{eqn:depart0}
\begin{split}
         \norm{\avgvect{w}_{t+1} - \vect{w}^*}^2 &  = \norm{\avgvect{w}_{t+1} - \avgvect{u}_{t+1} + \avgvect{u}_{t+1}- \vect{w}^*}^2 \\
        & = \underbrace{\norm{\avgvect{w}_{t+1} - \avgvect{u}_{t+1}}^2}_{A_1} + \underbrace{\norm{\avgvect{u}_{t+1}- \vect{w}^*}^2}_{A_2}  + \underbrace{2\dotp{\avgvect{w}_{t+1} - \avgvect{u}_{t+1}}{\avgvect{u}_{t+1}- \vect{w}^*}}_{A_3}.
\end{split}  
\end{equation*}
% \zixiangr{
The expectation of $A_1$ can be bounded using Lemma \ref{lemma:ROUL}.  We then bound the expectation of $A_3$ by the Cauchy–Schwarz inequality: 
\begin{equation*}
\begin{split}
    \expt  & \squab{2\dotp{\avgvect{w}_{t+1} -  \avgvect{u}_{t+1}}{\avgvect{u}_{t+1}- \vect{w}^*}} \leq \frac{2\sqrt{dD(\gamma)}}{K}\eta_t EH\expt\norm{\avgvect{u}_{t+1}- \vect{w}^*},
\end{split}
\end{equation*}
% }
and it is now related to $A_2$. Finally, we can write $A_2$ as
 \begin{equation*} \label{eqn:depart1}
 \begin{split}
         \norm{\avgvect{u}_{t+1} - \vect{w}^*}^2 & = \norm{\avgvect{u}_{t+1} - \avgvect{v}_{t+1} + \avgvect{v}_{t+1}- \vect{w}^*}^2 \\
        & = \underbrace{\norm{\avgvect{u}_{t+1} - \avgvect{v}_{t+1}}^2}_{B_1} + \underbrace{\norm{\avgvect{v}_{t+1}- \vect{w}^*}^2}_{B_2}  + \underbrace{2\dotp{\avgvect{u}_{t+1} - \avgvect{v}_{t+1}}{\avgvect{v}_{t+1}- \vect{w}^*}}_{B_3}.
    \end{split}
\end{equation*}
Based on Lemma~\ref{lemma:sample}, the expectation of $B_3$ over random user selection is zero, since we have $\expt \squab{\avgvect{u}_{t+1} - \avgvect{v}_{t+1}} = \vect{0}$. The expectation of $B_1$ can be bounded also by Lemma~\ref{lemma:sample}. Therefore, we have
\zixiangr{
   \begin{align*}\label{eqn:case3}\small
         \expt\norm{\avgvect{w}_{t+1}-\vect{w}^*}^2 & \leq \left (1 + \frac{2\sqrt{dD(\gamma)}}{K} \eta_t EH \right )\expt\norm{\avgvect{v}_{t+1}- \vect{w}^*}^2\\
        &  + \left (1 + \frac{2\sqrt{dD(\gamma)}}{K}\eta_t EH \right )
        \frac{M-K}{M-1} \frac{4}{K} \eta_t^2 E^2 H^2 +  \frac{4dD(\gamma)}{K^2} \eta^2_{t} E^2 H^2   \\
        & \leq (1-\eta_t \mu\zixiangbbb{^\prime}) \expt \norm{\avgvect{w}_t - \vect{w}^*}^2 + \eta_t^2 \left[ \left(1 + \frac{2\sqrt{dD(\gamma)}}{K}\eta_t EH \right )\right.\\
        & \times \left(\sum_{k=1}^M \frac{H_k^2}{M^2} + 6L \Gamma  + 8(E-1)^2 H^2 + \frac{M-K}{M-1} \frac{4}{K}  E^2 H^2 \right)\left. + {\frac{4dD(\gamma)}{K^2}  E^2 H^2} \right],  
    \end{align*}
where $\mu^\prime \triangleq \mu - 2\sqrt{dD(\gamma)}EH/K$. Let $\Delta_t \triangleq \expt \norm{\avgvect{w}_{t}- \vect{w}^*}^2$. No matter whether $t+1 \in \mathcal{I}_E$ or $t+1 \notin \mathcal{I}_E$, we always have
    \begin{equation*}
      \Delta_{t+1} \leq (1 -\eta_t\mu{^\prime}) \Delta_{t} + \eta_t^2 G,
    \end{equation*}
    where
    \begin{equation*}\small
    \begin{split}
      G \triangleq & \left (1 + \frac{2\sqrt{dD(\gamma)}}{K}EH\eta_1 \right ) \left(\sum_{k=1}^M \frac{H_k^2}{M^2} + 6L \Gamma + 8(E-1)^2 H^2 + \frac{M-K}{M-1} \frac{4}{K}  E^2 H^2 \right)  + {\frac{4dD(\gamma)}{K^2}  E^2 H^2}.
    \end{split}
    \end{equation*}
 Define $v\triangleq \max\{\frac{4G}{{\mu^\prime}^2}, (1 + r)\Delta_1\}$. By choosing $\eta_t = \frac{2}{\mu^\prime(t+r)} $, we can prove  $\Delta_t\leq\frac{v}{t + r}$ by induction:
 \begin{align*}
 \begin{split}
       \Delta_{t + 1} & \leq \left(1 - \frac{2}{t + r}\right)\Delta_t + \frac{4G}{{\mu^\prime}^2(t + r)^2}  \\
       &= \frac{t + r -2}{(t + r)^2}v + \frac{4G}{{\mu^\prime}^2(t + r)^2} \\
      & \leq \frac{t + r -1}{(t + r)^2}v + \frac{4G}{{\mu^\prime}^2(t + r)^2} - \frac{v}{(t + r)^2} \\
      & \leq \frac{v}{t + r + 1}.
 \end{split}
 \end{align*}
By the $L$-smoothness of $f$ and $v \leq \frac{4G}{{\mu^\prime}^2} + (1 + r)\Delta_1$, we prove Theorem \ref{thm.Conv}.}

\bibliographystyle{IEEEtran}
%\bibliography{wireless,Shen,ref,reference}
\bibliography{reference}

\end{document}